\documentclass[]{imsart}
\RequirePackage[numbers]{natbib}
\RequirePackage{amsthm,amsmath,amsfonts} %,amssymb}
\RequirePackage[colorlinks,citecolor=blue,urlcolor=blue]{hyperref}
\RequirePackage{graphicx}
\usepackage{mathbbol}
\DeclareSymbolFontAlphabet{\amsmathbb}{AMSb}%
\usepackage{parskip}
\usepackage{enumerate}   
\usepackage{amsmath}
\usepackage{amsthm}
\usepackage{amsfonts}
\usepackage{graphicx}
\usepackage{xcolor}
\usepackage{caption}
\usepackage{pstricks-add}
\usepackage{tikz}
\usepackage{color}
\usepackage{xparse}
\usepackage{subfigure}
\usepackage{multirow}
\usepackage{diagbox}
\arxiv{2512.14959}
\startlocaldefs
\newcommand{\z}{\boldsymbol{z}}
\newcommand{\uu}{\boldsymbol{u}}
\newcommand{\E}{\amsmathbb{E}}
\newcommand{\B}{\boldsymbol{B}}
\newcommand{\R}{\amsmathbb{R}}
\newcommand{\Z}{\boldsymbol{Z}}
\newcommand{\HH}{\amsmathbb{H}^{(n)}}
\newcommand{\HHd}{\amsmathbb{H}^{(n)}_{\dagger,1}}

\newcommand{\mm}{\mathbb{m}^{(n)}}
\newcommand{\md}{\mathbb{m}^{(n)}_{\dagger,1}}

\newcommand{\g}{\mathbb{g}^{(n)}}

\newcommand{\Fd}{\amsmathbb{F}_{\dagger}^{(n)}}
\newcommand{\G}{\amsmathbb{G}}

\newcommand{\LLd}{\mathbb{\Lambda}_{\dagger}^{(n)}}

\definecolor{clrgreen}{RGB}{50, 205, 50}

\definecolor{clrblue}{RGB}{1, 191, 255}

\definecolor{clrred}{RGB}{254, 3, 106}

\definecolor{clrorange}{RGB}{255, 165, 0}

\definecolor{clrgrey}{RGB}{0, 90, 124}

 \definecolor{clrbiasedtrue}{RGB}{255, 0, 0}
\newcommand{\truebiased}{\protect \tikz[baseline=-0.6ex]
 {\protect\draw[line width=1pt, color=clrbiasedtrue](0,0) -- (0.5,0.0)}}

 \definecolor{clrunbiasedtrue}{RGB}{0, 0, 0}
\newcommand{\trueunbiased}{\protect \tikz[baseline=-0.6ex]
 {\protect\draw[line width=1pt, color=clrunbiasedtrue](0,0) -- (0.5,0.0)}}

  \definecolor{clrestimate}{RGB}{0, 71, 171}
\newcommand{\estimate}{\protect \tikz[baseline=-0.6ex]
 {\protect\draw[line width=1pt, color=clrestimate](0,0) -- (0.5,0.0)}}
 
  \definecolor{naive}{RGB}{255, 165, 0}
\newcommand{\naive}{\protect \tikz[baseline=-0.6ex]
 {\protect\draw[line width=1pt, color=naive](0,0) -- (0.5,0.0)}}

  \definecolor{uniform}{RGB}{87, 185, 255}
\newcommand{\uniform}{\protect \tikz[baseline=-0.6ex]
 {\protect\draw[line width=1pt, color=uniform](0,0) -- (0.5,0.0)}}

   \definecolor{specific}{RGB}{57, 173, 72}
\newcommand{\specific}{\protect \tikz[baseline=-0.6ex]
 {\protect\draw[line width=1pt, color=specific](0,0) -- (0.5,0.0)}}

\newcommand\independent{\protect\mathpalette{\protect\independenT}{\perp}}
\def\independenT#1#2{\mathrel{\rlap{$#1#2$}\mkern2mu{#1#2}}}

\theoremstyle{plain}

\newtheorem{theorem}{Theorem}[section]
\newtheorem{lemma}[theorem]{Lemma}

\theoremstyle{definition}
\newtheorem{definition}[theorem]{Definition}
\newtheorem{example}{Example}

\theoremstyle{remark}

\theoremstyle{plain}
\newtheorem{corollary}[theorem]{Corollary}

\theoremstyle{definition}
\newtheorem{assumptions}[theorem]{Assumptions}
\newtheorem{remark}[theorem]{Remark}
\endlocaldefs

\begin{document}
\begin{frontmatter}
\title{Nonparametric Survival Estimation with Contaminated and Adjudicated Events}
%\title{A sample article title with some additional note\thanksref{t1}}
\runtitle{Conditional Expert Kaplan--Meier Estimation}
%\thankstext{T1}{A sample additional note to the title.}

\begin{aug}
%%%%%%%%%%%%%%%%%%%%%%%%%%%%%%%%%%%%%%%%%%%%%%%
%% Only one address is permitted per author. %%
%% Only division, organization and e-mail is %%
%% included in the address.                  %%
%% Additional information can be included in %%
\author[A]{\fnms{Martin}~\snm{Bladt}\ead[label=e1]{martinbladt@math.ku.dk}} \and
\author[A]{\fnms{Kristian Vilhelm}~\snm{Dinesen}\ead[label=e2]{kristian2904@hotmail.com}}
%%%%%%%%%%%%%%%%%%%%%%%%%%%%%%%%%%%%%%%%%%%%%%
%% Addresses                                %%
%%%%%%%%%%%%%%%%%%%%%%%%%%%%%%%%%%%%%%%%%%%%%%
\address[A]{Department of Mathematical Sciences, University of Copenhagen\printead[presep={,\ }]{e1,e2}}

% \address[B]{Department,
% University or Company Name\printead[presep={,\ }]{e2,e3}}
\runauthor{M. Bladt et al.}
\end{aug}

\begin{abstract}
We study the conditional expert Kaplan–Meier estimator, an extension of the classical Kaplan–Meier estimator designed for time-to-event data subject to both right-censoring and contamination. Such contamination, where observed events may not reflect true outcomes, is common in applied settings, including insurance and credit risk, where expert opinion is often used to adjudicate uncertain events. Building on previous work, we develop a comprehensive asymptotic theory for the conditional version incorporating covariates through kernel smoothing. We establish functional consistency and weak convergence under suitable regularity conditions and quantify the bias induced by imperfect expert information. The results show that unbiased expert judgments ensure consistency, while systematic deviations lead to a deterministic asymptotic bias that can be explicitly characterized. We examine finite-sample properties through simulation studies and illustrate the practical use of the estimator with an application to loan default data.
\end{abstract}

\begin{keyword}[class=MSC]
\kwd[Primary ]{62N02}
\kwd{62G08}
%\kwd[; secondary ]{00X00}
\end{keyword}

\begin{keyword}
\kwd{Right-censoring}
\kwd{expert information}
\kwd{multivariate kernel smoothing}
\kwd{weak convergence}
\end{keyword}

\end{frontmatter}
%%%%%%%%%%%%%%%%%%%%%%%%%%%%%%%%%%%%%%%%%%%%%%
%% Please use \tableofcontents for articles %%
%% with 50 pages and more                   %%
%%%%%%%%%%%%%%%%%%%%%%%%%%%%%%%%%%%%%%%%%%%%%%
%\tableofcontents

\section{Introduction}
Statistical analyses of right-censored time-to-event data typically assume that event times, once they occur, are correctly observed. In many applied settings, however, this assumption fails due to contamination, where recorded events do not reliably correspond to true outcomes. If left untreated, such contamination can introduce bias into estimates of the underlying event-time distribution. This challenge is commonly referred to as inference under incomplete event adjudication. To address it, we study how expert information may be incorporated in a nonparametric framework.

Expert information appears in settings where statistical analyses require explicit modelling choices. Bayesian procedures depend on the specification of priors, while frequentist methods require decisions regarding models, smoothing parameters, and test statistics. These choices introduce elements of subjectivity, even if this is not always highlighted in applied work. Frameworks for incorporating expert information aim to make such choices explicit and to clarify when expert assessments can improve estimation; see~\cite{ExpertJudgmentsSSI, ExpertJudgmentBrownstein}. Building on the Kaplan–Meier estimator from~\cite{KM}, and its covariate extensions in~\cite{Beran, Dabrowska}, we introduce a conditional Kaplan–Meier estimator that integrates expert judgments to enable estimation under contamination.

Contamination arises in many real-world settings. Our motivation is drawn from applications in insurance and credit risk, where complex data-collection processes and technical event definitions can complicate the identification of true events. For example, the European Banking Association (EBA) defines loan default using a detailed set of criteria that place stringent demands on data completeness and quality, making reliable adjudication difficult in practice. In disability insurance, a claim may be incorrectly closed and later reopened, resulting in a longer disability duration and larger payout than initially recorded. Practitioners are often aware of such contamination mechanisms, yet lack quantitative tools to systematically incorporate this knowledge into statistical estimation.

Mathematically, our work builds upon the expert Kaplan–Meier estimator of~\cite{CrudeExpert}. We provide several extensions geared toward practical use and establish asymptotic theory beyond the usual functional consistency available for the standard Kaplan–Meier estimator. First, we incorporate potentially multivariate covariates, making the estimator more applicable to modern data environments. We then establish both functional consistency and weak convergence under settings with perfect and imperfect expert judgments. For the latter, we explicitly characterize the bias introduced by imperfect judgments and propose methods to assess how the quality of expert information influences estimation.

Several strands of literature relate to our contribution. The use of survival analysis to estimate loan default dates back to~\cite{Narain} and has since been developed in multiple works. Classical survival models, including proportional hazards and accelerated failure time models, have been found competitive with industry-standard logistic regression~\cite{EuropeanGenericScoring}. Time-to-event approaches also enable more nuanced analyses, for example in behavioural scoring and profit expectation modelling; see~\cite{PHABScores, PersonalLoanData}. From an actuarial standpoint,~\cite{Bladt02072020,ALBRECHER2024106171} propose combining expert information into a parametric framework for modeling non-life insurance claim sizes. Similarly,~\cite{SandquistIncomplete} studies estimation of conditional hazards in disability insurance within a multistate framework accommodating censoring, left-truncation, contamination, and reporting delays, proposing a parametric approach aligned with multistate models in life insurance. More broadly,~\cite{Cook2000} and~\cite{Cook2004} consider settings where clinical trial endpoints are weighted by their probability of later adjudication, reflecting related challenges in medical applications.

The remainder of the paper is structured as follows. Section~\ref{section_setup} introduces the set-up and motivates the conditional expert Kaplan–Meier estimator. Sections~\ref{section_consistency} and~\ref{section_normality} contain the consistency and weak convergence results, respectively. These theoretical sections also include practical examples illustrating expert-based estimation. Section~\ref{section_numerical} presents numerical studies supporting the asymptotic theory under both perfect and imperfect expert information and demonstrates the method using real-world loan default data. Section~\ref{section_conclusion} concludes. All proofs are found in the appendix.
\section{The conditional expert Kaplan--Meier estimator} \label{section_setup}

The starting point of this section is a brief review of the standard conditional Kaplan–Meier estimator, also called the Beran estimator. We recall its construction, as this provides the basis for incorporating expert judgments on contamination in a systematic way. We then formalize event contamination within a nonparametric framework for right-censored data and use this formulation to define the conditional expert Kaplan–Meier estimator, which integrates expert assessments directly into the estimation procedure.

\subsection{Preliminaries}
Let $(\Omega, \mathcal{A}, P)$ be a probability space and let $X\sim F$ be a random variable with values in $[0,\infty)$ such that $F(0)=0$. This is the variable of interest, and the aim is to estimate $X$. 
Let $C\sim J$ be a random variable with values in $[0,\infty]$ representing the censoring mechanism. We assume that $X$ and $C$ are dependent on a possibly high-dimensional variable $\boldsymbol{Z}\in \amsmathbb{R}^{k}$ with density $g$. We then observe 
\begin{align*}
    W=X\wedge C,\quad \delta = 1_{(W=X)}, \quad \boldsymbol{Z}.
\end{align*}
Let $(W_{i},\delta_{i}, \boldsymbol{Z}_{i})$ be a sequence of iid copies of the triplet $(W, \delta, \boldsymbol{Z})$ according to this specification. Under the assumption of conditionally independent (entirely random) right-censoring, meaning that $X\underset{\boldsymbol{Z}}{\independent} C$, this constitutes the well-known statistical problem of inference under right-censoring, in which the usual conditional Kaplan--Meier estimator is relevant. The estimator may be derived as follows. Let
\begin{align*}
    F(t|\z) = P(X\leq t \mid \boldsymbol{Z}=\z)=\int_{[0,t]}1-F(s-|\z)\,\mathrm{d}\Lambda(s|\z),
\end{align*}
where 
\begin{align*}
\Lambda(t|\z)=\int_{[0,t]}\frac{1}{1-F(s-|\z)}\,\mathrm{d}F(s|\z), \quad 0\leq t <F^{-1}(1|\z),
\end{align*}
is the cumulative hazard function. Let $H(t|\z)=P(W\leq t\mid \boldsymbol{Z}=\z)$. Under conditionally independent right-censoring it holds that $(1-F(t|\z))(1-J(t|\z))=1-H(t|\z)$, and the following expression is obtained:
\begin{align}
    \label{cumulative2}\Lambda(t|\z)=\int_{[0,t]}\frac{1}{1-H(s-|\z)}\,\mathrm{d}H_{1}(s|\z), \quad 0\leq t < H^{-1}(1|\z),
\end{align}
where $H_{1}(t|\z)=P(W\leq t, W=X\mid \boldsymbol{Z}=\z)$. Define the estimators $\amsmathbb{H}^{(n)}(\cdot|\z)$ and $\amsmathbb{H}^{(n)}_{1}(\cdot|\z)$ of $H(\cdot|\z)$ and $H_{1}(\cdot|\z)$, respectively, by the Nadaraya--Watson type regression estimators
\begin{align*}
    \amsmathbb{H}^{(n)}(t|\z)=\frac{\mm(t;\z)}{\g(\z)},\quad 
    \amsmathbb{H}_{1}^{(n)}(t|\z)=\frac{\mm_{1}(t;\z)}{\g(\z)},
\end{align*}
where
\begin{align*}
    \mm(t;\z)&=\sum_{i=1}^{n}1_{(W_{i}\leq t)}K_{\B_{n}}\left(\Z_{i}-\z\right), \\ 
    \mm_{1}(t;\z)&=\sum_{i=1}^{n}1_{(W_{i}\leq t)}\delta_{i}K_{\B_{n}}\left(\Z_{i}-\z\right), \\
    \g(\z)&=\sum_{i=1}^{n}K_{\B_{n}}\left(\Z_{i}-\z\right)
\end{align*}
for 
\begin{align*}
    K_{\B_{n}}\left(\boldsymbol{Z}-\z\right) = \frac{1}{n|\B_{n}|}K\left(\boldsymbol{B}_{n}^{-1}(\boldsymbol{Z}-\z)\right).
\end{align*}
Here, $K\colon\amsmathbb{R}^{k}\to \amsmathbb{R}_{\geq 0}$ is a bounded kernel function with compact support, and $(\B_{n})$ is a sequence of $k\times k$ non-singular bandwidth matrices with non-negative entries. It is assumed throughout that 
\begin{align*}
    \int K(\uu)\,\mathrm{d}\uu=1,\quad 
    \int \uu^{T}K(\uu)\,\mathrm{d}\uu=0,\quad 
    |\B_{n}|\to 0,\quad 
    n|\B_{n}|\to \infty,
\end{align*}
and that $\boldsymbol{z}\in S\subseteq \R^{k}$ with $P(\boldsymbol{Z}\in S)=1$. Using the estimators $\amsmathbb{H}^{(n)}(\cdot|\z)$ and $\amsmathbb{H}_{1}^{(n)}(\cdot|\z)$ we may obtain an empirical version of the cumulative hazard in~\eqref{cumulative2} as 
\begin{align*}
   \mathbb{\Lambda}^{(n)}(t|\z)=\int_{[0,t]}\frac{1}{1-\amsmathbb{H}^{(n)}(s-|\z)}\,\mathrm{d}\amsmathbb{H}_{1}^{(n)}(s|\z), 
   \quad 0\leq t < H^{-1}(1|\z).
\end{align*}
The conditional Kaplan--Meier estimator $\amsmathbb{F}^{(n)}(\cdot|\z)$ is then obtained by evaluating the estimator of the cumulative hazard in the product integral $\phi$:
\begin{align*}
    1-\amsmathbb{F}^{(n)}(t|\z)
    =\phi\left(-\mathbb{\Lambda}^{(n)}(t|\z)\right)
    =\prod_{s\leq t}\left(1-\Delta \mathbb{\Lambda}^{(n)}(s|\z)\right).
\end{align*}

% The product integral links cumulative hazard functions to survival probabilities for probability distributions, which is the important role of the product integral for our scope. The product integral may defined several ways more general than the above, for example through a certain Volterra integral equation, the Péano series definition and the product-limit definition to mention a few, confer with \cite{ProductInt}. The product integral has many interesting properties and we shall rely heavily on a certain \textit{Duhamel} equation for product integrals in the following sections on asymptotics.
% 
% The estimators $\mm(\cdot;\z)$ and $\mm_{1}(\cdot;\z)$ and the estimator $\f(\z)$ are the Preistly-Chao and Parzen-Rosenblatt estimators, respectively, that make up the Nadaraya-Watson type of estimators $\amsmathbb{H}^{(n)}(\cdot|\z)$ and $\amsmathbb{H}_{1}^{(n)}(\cdot|\z)$. These non-parametric estimators come in a variety of different forms, e.g. the Nadaraya-Watson estimator is a special case of the so-called local polynomials which facilitate a broader family of regression estimators, confer with \cite{Härdle, KS}.
% 
In the sections below on functional asymptotic theory, we work in the space $\ell^{\infty}(\mathcal{F}, |\cdot|{\infty})$, the set of uniformly bounded real-valued functions on $\mathcal{F}$. A function $z : \mathcal{F} \to \amsmathbb{R}$ belongs to this space when $|z|{\infty} = \sup_{f \in \mathcal{F}} |z(f)| < \infty$. For background on weak convergence and empirical process theory, we refer to~\cite{WCEP, AS, EMP}.

\subsection{Survival analysis under contamination}
We now extend the set-up to include a contamination mechanism. The contamination mechanism is a random variable $Y$ with values in $[0,\infty]$. We observe iid data points $(W_{i},\delta_{i},\boldsymbol{Z}_{i})$ according to
\begin{align}
    W=X\wedge (Y\wedge C), \quad \delta=1_{\{W=(X\wedge Y)\}}, \quad \boldsymbol{Z}. \label{contamination_setup}
\end{align}
The contamination variable $Y$ has the effect that the event $(W=X)$ is no longer directly accessible through $\delta$. This formalizes the setting in which biased estimates arise in practice: under contamination, it may occur that $\delta = 1$ while at the same time $Y<X$. Such an observation is referred to as a contaminated or false event. If $\delta$ is used without caution on contaminated data and incorrectly interpreted as the indicator of the event $(W=X)$, this leads to biased estimates that exaggerate the true hazard rate. The case $\delta=1$ and $X\leq Y$ is referred to as a true event, and whenever $\delta=0$ the event time is censored. 

To handle such ambiguous observations, the expert provides a sequence of random variables intended to discriminate contaminated events from truly observed events:
\begin{definition}[The expert] \label{def_expert_conditional_km}
Let $(W_{nm}, \delta_{nm}, \boldsymbol{Z}_{nm})$ be iid copies of $(W,\delta,\boldsymbol{Z})$ according to~\eqref{contamination_setup}, indexed by a triangular array. The expert provides a triangular array of random variables
\begin{align*}
    \eta_{n1},\dots,\eta_{nm},\dots,\eta_{nn}, \quad \text{for } 1\leq m\leq n \text{ and } n\in \amsmathbb{N},
\end{align*}
where $\eta_{nm}\in \{0,1\}$ represents the expert's best judgment of the unobservable event indicator $1_{(W_{nm}=X_{nm})}$.
\end{definition}

\indent
As usual for a triangular array, we assume that for every $n$ the random variables $\eta_{n1},\ldots,\eta_{nn}$ are mutually independent but not necessarily identically distributed. The expert judgments are exogenously given, and all implementation-specific details concerning these judgments are abstracted away in the definition. This provides flexibility when designing expert judgments for practical estimation purposes. We emphasize that any covariate may be used in forming the expert judgments to adjudicate the uncertain events.

Similar to the usual conditional Kaplan--Meier setting, we say that conditional contaminated right-censoring is entirely random if 
\[
    X\underset{\boldsymbol{Z}}{\independent}(Y\wedge C).
\]
Under this assumption it holds that 
\begin{align}
    \label{cumulative}\Lambda(t|\z)=\int_{[0,t]}\frac{1}{1-H(s-|\z)}\,\mathrm{d}H_{1}(s|\z), \quad 0\leq t < H^{-1}(1|\z).
\end{align}

We replace the observables $(\delta_{nm})$ by the expert judgments $(\eta_{nm})$ in the Priestley--Chao estimator from above and introduce the estimator $\amsmathbb{H}_{\dagger,1}^{(n)}(\cdot|\z)$ of $H_{1}(\cdot|\z)$ as
\begin{align*}
    \amsmathbb{H}^{(n)}_{\dagger,1}(t|\z)
    =\frac{\mm_{\dagger,1}(t;\z)}{\g(\z)}
    =\frac{1}{\g(\z)}\sum_{m=1}^{n}1_{(W_{nm}\leq t)}\eta_{nm}K_{\B_{n}}\left(\Z_{nm}-\z\right),
    \quad t\geq 0.
\end{align*}
The subscript ``$\dagger$'' is used for estimators depending on expert judgments. We arrive at the following definition:
\begin{definition}[The conditional expert Kaplan--Meier estimator]
Let $(W_{nm}, \delta_{nm}, \boldsymbol{Z}_{nm})$ be iid copies of $(W,\delta,\boldsymbol{Z})$ according to~\eqref{contamination_setup}, indexed by a triangular array. The conditional expert Kaplan--Meier estimator $\amsmathbb{F}^{(n)}_{\dagger}(\cdot|\z)$ of $F(\cdot|\z)$ is given by 
\begin{align*}
    1-\amsmathbb{F}_{\dagger}^{(n)}(t|\z)
    =\phi\left(-\mathbb{\Lambda}_{\dagger}^{(n)}(t|\z)\right)=\prod_{s\leq t}\left(1-\Delta \mathbb{\Lambda}^{(n)}_{\dagger}(s|\z)\right),
\end{align*}
where $\mathbb{\Lambda}_{\dagger}^{(n)}(\cdot|\z)$ is the conditional expert Nelson--Aalen estimator of $\Lambda(\cdot|\z)$ given by 
\begin{align*}
    \mathbb{\Lambda}_{\dagger}^{(n)}(t|\z)
    =\int_{[0,t]}\frac{1}{1-\amsmathbb{H}^{(n)}(s-|\z)}\,\mathrm{d}\amsmathbb{H}_{\dagger,1}^{(n)}(s|\z),
    \quad 0\leq t < H^{-1}(1|\z).
\end{align*}
\end{definition}

We make no a priori assumptions about the expert's ability to identify and censor contaminated events. This framework accommodates expert judgments of varying quality. In the subsequent sections on consistency and weak convergence, we establish conditions under which the resulting survival estimates are asymptotically unbiased. We further develop the asymptotic theory to quantify the bias that arises when the expert's judgments do not meet these conditions.

\section{Functional consistency}\label{section_consistency}
\noindent This section presents results on the consistency of the conditional expert Kaplan--Meier estimator $\amsmathbb{F}_{\dagger}^{(n)}(\cdot|\z)$, incorporating both perfect and imperfect expert judgments. We build on~\cite{Stute}, which establishes almost sure convergence of a kernel regression estimator compatible with the framework described above. Specifically, the results that follow are derived under Assumptions~\ref{Strong_uniform_assumptions_2}.
\begin{assumptions}[Functional consistency] \label{Strong_uniform_assumptions_2} Assume that 
    \begin{enumerate}
        \item $K$ is the product kernel, i.e. $K(\boldsymbol{u})=\prod_{i=1}^{k}\Tilde{K}(u_{i})$ for a kernel $\Tilde{K}\colon\amsmathbb{R}\to \amsmathbb{R}$ \label{strong_uniform_consistency_2_i}
         % \item $\Tilde{K}$, and hence $K$ has bounded support
        \item $\boldsymbol{B}_{n}$ is diagonal for all $n$ with $\boldsymbol{B}_{n}(i,i)=b_{n}>0$ for all $1\leq i\leq k$\label{strong_uniform_consistency_2_ii}
        % \item $b_{n}\to 0$
        % and $nb_{n}^{k}\to \infty$
        \item    \label{strong_uniform_consistency_2_iii}$\sum_{n=1}^{\infty}c_{n}^{r}<\infty$ for some $r>1$ and where $c_{n}=\log(n)(nb_{n}^{k})^{-1}$
        \item \label{strong_uniform_consistency_2_iv}         $\sum_{n=1}^{\infty}2\exp\left(-\frac{2\varepsilon^{2} n|\B_{n}|^{2}}{\sup_{\uu}K^{2}(\uu)}\right)<\infty$ for every $\varepsilon>0$\label{consistency2_assumptions}
        \item $g$ is $C^{2}$.
    \end{enumerate}
\end{assumptions}
\begin{theorem}[Functional consistency of the conditional expert Kaplan--Meier estimator] \label{consistency_theorem}Let Assumptions \ref{Strong_uniform_assumptions_2} hold. Assume that conditional contaminated right-censoring is entirely random and that 
\begin{align}
        \bigg\vert\E\left[\md(t;\boldsymbol{z})\right]-H_{1}(t|\z)g(\z)\bigg \vert \to 0,\label{consistency_assumption}
    \end{align}
    for every $t\geq 0.$
    For $0\leq \theta< H^{-1}(1|\z)$ it then holds that 
\begin{align*}
    \sup_{0\leq t \leq \theta}|\mathbb{\Lambda}_{\dagger}^{(n)}(t|\z)-\Lambda(t|\z)|&\overset{\textnormal{a.s.}}{\to}0
    \\
    \sup_{0\leq t \leq \theta}|\Fd(t|\z)-F(t|\z)|&\overset{\textnormal{a.s.}}{\to}0.
\end{align*}
\end{theorem}
% \begin{proof}
%     This proof is virtually identical to the proof of Theorem 7.3.1 in \cite{EMP} or the proof of Theorem 3.4 in \cite{CrudeExpert}. We provide the details in the appendix for completeness. 
% \end{proof}
\begin{remark}\label{bandwidth_remark_consistency}
 Assumptions~\ref{Strong_uniform_assumptions_2} are fulfilled by any product kernel and the following choice of bandwidth 
    \begin{align*}
        b_{n}=\left(\frac{\log(n)}{n^{\rho}}\right)^{1/k}
    \end{align*}
    for $0<\rho<1/2$. Sufficient conditions for the assumption in~\eqref{consistency_assumption} to hold is to have iid unbiased expert judgments 
    \begin{align*}
        \eta_{i}\sim \delta_{i}\,\text{Bern}(p(W_{i},\boldsymbol{Z}_{i})),
    \end{align*}
    where $p$ is a function given by
    \begin{align*}
        p(w,\z)=\E\left[1_{(W=X)}|\delta=1,W=w,\boldsymbol{Z}=\z\right]
        %=P(X\leq Y|\delta=1,W=w,\boldsymbol{Z}=\z),
    \end{align*}
    and the condition that $\z\mapsto H_{1}(t|\z)$ is $C^{2}$ for every $t\geq 0$. A full justification for this remark is given in the appendix.
    % satisfies
    % \begin{align*}
    %     p(w,\z)=\E[1_{(W=X)}|\boldsymbol{Z}=\z]
    % \end{align*}
    % for every $w\geq 0$.
    % to hold is have iid unbiased judgments $(\eta_{i})$ and the conditions 
    % \begin{align*}
    %     \E[\eta|\boldsymbol{Z}=\z]=\E[1_{(W=X)}|\boldsymbol{Z}=\z],\quad \z\mapsto H_{1}(t|\z)\text{ is }C^{2}\text{ for every }t\geq 0,
    % \end{align*}
    % where $\eta\overset{\text{D}}{=}\eta_{1}$. 
\end{remark}
Unbiased expert judgments are sufficient for consistency under contamination. For example, the expert distribution may be defined using a suitable maximum likelihood estimator based on historical data where the effects of contamination have been observed over time. Theorem~\ref{consistency_theorem}, together with Remark~\ref{bandwidth_remark_consistency}, provides guidance on handling contamination in practice.

Consider a simple scenario with uniform contamination, where
\begin{align*}
\E[1_{(W=X)}|\delta=1,W=w,\boldsymbol{Z}=\z]=\E[1_{(W=X)}|\delta=1]=c
\end{align*}
for some $c \in [0,1]$. In this case, an expert censoring observed events with probability $1-c$ ensures consistency of the estimator. Uniform contamination patterns are realistic in some applications; for instance, in disability insurance, claims may be equally susceptible to processing or health-examination errors regardless of individual characteristics.

Such scenarios, however, presuppose the availability of historical data or other means to access contamination effects. In many practical situations, complete information about the contamination mechanism is not available, and expert judgments may therefore be biased. It is important to examine how such bias impacts asymptotic consistency and to quantify its effect for specific implementations of expert judgments. This analysis is presented below.
\begin{corollary}[Bias of the conditional expert Kaplan--Meier estimator] \label{bias}
    Let Assumptions~\ref{Strong_uniform_assumptions_2} hold and assume that conditional contaminated right-censoring is entirely random. Let $\gamma(\cdot;\z)$ be a bias function of the expert meaning that 
    \begin{align*}
        \bigg\vert\E\left[\md(t;\boldsymbol{z})\right]-\left(H_{1}(t|\z)g(\z) + \gamma(t;\z)\right) \bigg\vert \to 0 
    \end{align*}
    for every $t\geq 0$. Then it holds that 
\begin{align*}
    |\Fd(t|\z)-\phi\left(-\Lambda(t|\z)-\Gamma(t;\z)\right)|\overset{\textnormal{a.s.}}{\to} 0
\end{align*}
for $0\leq t<H^{-1}(1|\z)$, where
\begin{align*}
    \Gamma(t;\z)=\int_{[0,t]}\frac{1}{g(\z)\left(1-H(s-|\z)\right)}\textnormal{d}\gamma(s;\z).
\end{align*}
If it holds that $H_{1}(\cdot|\z)g(\z)+\gamma(\cdot;\z)$ is a non-decreasing and bounded càdlàg function, then for $0\leq \theta<H^{-1}(1|\z)$ the convergence is in $\ell^{\infty}([0,\theta])$ 
\begin{align*}
    \sup_{0\leq t \leq \theta}|\Fd(t|\z)-\phi\left(-\Lambda(t|\z)-\Gamma(t;\z)\right)|&\overset{\textnormal{a.s.}}{\to}0.
\end{align*}
\end{corollary}
\begin{example}[Example on biased judgments] 
%  Consider the function $p$ given by
% \begin{align*}
% p(w,\z)=\E\left[1_{(W=X)}\mid \delta=1, W=w, \boldsymbol{Z}=\z\right] = P(X\leq Y \mid \delta=1, W=w, \boldsymbol{Z}=\z).
% \end{align*}
\label{example_consistency}
We extend the setup of~\cite{CrudeExpert} by considering expressions from which expert judgments may be obtained using the function $p$ from Remark \ref{bandwidth_remark_consistency} above. In the case of full information about the contamination mechanism, which is equivalent to knowledge of $p$, we define the perfect expert ensuring consistency by the iid judgments
\begin{align*}
    \eta^{\star}_{i} \sim \delta_{i}\,\text{Bern}\left(p(W_{i},\boldsymbol{Z}_{i})\right).
\end{align*}
 We define the naïve expert as the expert taking all observed events as true; that is, the naïve expert provides the iid judgments
\begin{align*}
    \eta^{\times}_{i} = \delta_{i}.
\end{align*}
For inference purposes, a natural approach is to use an estimated function $\hat{\mathbb{p}}$ of $p$ to create expert judgments. We consider the practically feasible expert to hold partial information about the true contamination mechanism $p$, somewhere between the naïve and perfect expert. To this end, assume that the estimated $\hat{\mathbb{p}}$ can be expressed as
\begin{align*}
    \hat{\mathbb{p}}(w,\z) = (1-p_{0}) + p_{0}\,p(w,\z),
\end{align*}
where $p_{0}\in [0,1]$ is an information parameter bridging naïve and perfect judgments. The practically feasible expert then provides the iid judgments
\begin{align*}
     \eta_{i} \sim \delta_{i}\,\text{Bern}\left(\hat{\mathbb{p}}(W_{i},\boldsymbol{Z}_{i})\right).
\end{align*}
Casting expert judgments in this way facilitates a natural way to quantify expert bias.

Using Lemma~\ref{weak_consistency_proof}, under differentiability assumptions on event distributions, the following weak convergence limit is obtained for the practically feasible expert:
\begin{align*}
\E\left[\md(t;\z)\right] \to (1-p_{0}) H_{1}^{\times}(t|\z) g(\z) + p_{0} H_{1}(t|\z) g(\z), \quad t\geq 0,
\end{align*}
where $H_{1}^{\times}(t|\z) = P(W\leq t, \delta=1 \mid \boldsymbol{Z}=\z)$. This implies a bias $\gamma(\cdot;\z)$ given by 
\begin{align*}
    \gamma(t;\z) = (1-p_{0})  \left(H_{1}^{\times}(t|\z) - H_{1}(t|\z)\right)g(\z), \quad t\geq 0,
\end{align*} 
which naturally decreases as a function of the information parameter $p_{0}$. This bias can be further characterized in terms of the function $p$ and $H_{1}^{\times}(\cdot|\z)$. It holds that 
\begin{align*}
H_{1}^{\times}(t|\z) - H_{1}(t|\z) = \int_{[0,t]} 1 - p(s,\z)\mathrm{d}H^{\times}_{1}(s|\z), \quad t\geq 0,
\end{align*}
and $\Gamma(\cdot;\z)$ is then determined as
\begin{align*}
    \Gamma(t;\z) = (1-p_{0}) \int_{[0,t]} \frac{1-p(s,\z)}{1-H(s-|\z)} \mathrm{d}H_{1}^{\times}(s|\z), \quad 0 \leq t < H^{-1}(1|\z).
\end{align*}

Assume for simplicity that $\Lambda(\cdot|\z)$ and $\Gamma(\cdot;\z)$ are continuous. By Corollary~\ref{bias} we then have  
\begin{align*}
    \Fd(\cdot|\z) \overset{\text{a.s.}}{\to} \text{e}^{-\Gamma(\cdot;\z)} F(\cdot|\z) \quad \text{in } \ell^{\infty}([0,\theta]).
\end{align*}
Here, $\text{e}^{-\Gamma(\cdot;\z)}$ serves as a factor expressing how biased expert judgments manifest as bias in the conditional expert Kaplan--Meier estimator. This factor can be approximated by plugging in the estimators of $p$ and $H(\cdot|\z)$ alongside an estimator of the naïve distribution $H_{1}^{\times}(\cdot|\z)$, obtainable directly from observed data. By benchmarking expert quality through $p_{0}$, it is then possible to evaluate how the resulting inference is influenced.

\end{example}
\section{Functional weak convergence}\label{section_normality}
We study the functional weak convergence of the conditional expert Kaplan–Meier estimator for two main reasons. Theoretically, weak convergence guarantees that estimation behaves in a predictable and well-understood manner as sample size increases, even in the presence of contamination and expert judgments. Practically, it enables the construction of approximate confidence intervals and hypothesis tests based on the asymptotic normal distribution. Beyond establishing weak convergence under sufficient expert information, it is also important to examine how imperfect information affects the limiting distribution. By parameterizing the quality of expert judgments, as in Example~\ref{example_consistency}, we explicitly relate the quality of expert input to the asymptotic bias of the survival estimates, providing a framework for uncertainty quantification in applied settings.

We impose the following assumptions on the underlying distributions and kernel behavior for the following results:
\begin{assumptions}[Weak convergence] \label{weak_convergence_assumptions_2}Assume that
\begin{enumerate}
     % \item $|\B_{n}|\to 0$ and $n|\B_{n}|\to \infty$
     \item $\sqrt{n|\B_{n}|}\left(\max_{i,j}\B_{n}(i,j)\right)^{2}\to 0$
     \item  $g$ is $C^{2}$ 
     \item $\z \mapsto H(t|\z)$ and $\z \mapsto H_{1}(t|\z)$ are $C^{2}$ for every $t\geq 0$.
\end{enumerate}
\end{assumptions}
% Under Assumptions \ref{weak_convergence_assumptions_2} and the assumption in \eqref{weak_convergence_condition} below we have Lemma \ref{weakconvergence_lemma} in hand from which the weak limit fidis the empirical processes corresponding to $\amsmathbb{H}^{(n)}(\cdot|\z)$ and $\amsmathbb{H}_{\dagger,1}(\cdot|\z)$ are determined. Using theory on stochastic convergence, we expand this convergence into functional weak convergence of the conditional expert Kaplan-Meier estimator:
% In Lemma (xx) this convergence in expanded into joint functional weak convergence, i.e.
% \begin{align*}
%     \left\{\sqrt{n|\B_{n}|}(\amsmathbb{H}^{(n)}(\cdot|\z) - H(\cdot|\z)), \sqrt{n|\B_{n}|}(\amsmathbb{H}^{(n)}_{\dagger,1}(\cdot|\z) - H_{1}(\cdot|\z))\right\}\overset{\text{D}}{\to}\left\{\amsmathbb{G}(\cdot|\z),\amsmathbb{G}_{1}(\cdot|\z)\right\},
% \end{align*}
% in $\ell^{\infty}[0,\infty)^{2}$, where $\{\amsmathbb{G}(\cdot|\z),\amsmathbb{G}_{1}(\cdot|\z)\}$ is a zero-mean Gaussian Process. From this, the weak convergence of the conditional crude expert Nelson-Aalen and Kaplan-Meier estimator follows by applications of the Functional Delta Theorem and we arrive at the following theorem:
Both the bandwidth assumption in Assumption~\ref{weak_convergence_assumptions_2} and the expert quality condition in equation~\eqref{weak_convergence_condition} below are satisfied by similar considerations as in Remark~\ref{bandwidth_remark_consistency} for $k<4$. For instance, for $k=1$ one may obtain $1/5<\rho<1/2$.
\begin{theorem}[Functional weak convergence of the conditional expert Kaplan--Meier estimator]  \label{weakconvergence} Let Assumptions~\ref{weak_convergence_assumptions_2} hold. Assume that conditional contaminated right-censoring is entirely random and that 
\begin{align} \label{weak_convergence_condition}
    \sqrt{n|\boldsymbol{B}_{n}|}\bigg\vert\E\left[\md(t;\boldsymbol{z})\right]-H_{1}(t|\z)g(\z)\bigg\vert \to 0
\end{align}
for every $t\geq 0$. For $0\leq \theta<H^{-1}(1|\boldsymbol{z})$ it then holds that 
    \begin{align*}
        \sqrt{n|\boldsymbol{B}_{n}|}\left(\LLd(\cdot |\boldsymbol{z})-\Lambda(\cdot |\boldsymbol{z})\right)&\overset{\textnormal{D}}{\to}\amsmathbb{L}(\cdot|\boldsymbol{z}),\quad 
        \sqrt{n|\boldsymbol{B}_{n}|}\left(\Fd(\cdot |\boldsymbol{z})-F(\cdot |\boldsymbol{z})\right)\overset{\textnormal{D}}{\to}\amsmathbb{Z}(\cdot |\boldsymbol{z})
    \end{align*}
    in $\ell^{\infty}([0,\theta])$. Here $\amsmathbb{L}(\cdot |\boldsymbol{z})$ and $\amsmathbb{Z}(\cdot|\z)$ are zero-mean Gaussian processes with covariance functions 
  \begin{align*}            \sigma_{\amsmathbb{L}(\cdot|\z)}^{2}(t,s)&=\frac{1}{g(\boldsymbol{z})}\int K^{2}(\boldsymbol{u})\textnormal{d}\boldsymbol{u}\int_{[0,s\wedge t]}\frac{1-\Delta \Lambda(u|\boldsymbol{z})}{\Bar{H}(u|\boldsymbol{z})}\textnormal{d}\Lambda(u|\boldsymbol{z}),
  \\
  \sigma^{2}_{\amsmathbb{Z}(\cdot|\z)}(t,s)&=\frac{\Bar{F}(t|\boldsymbol{z})\Bar{F}(s|\boldsymbol{z})}{g(\boldsymbol{z})}\int K^{2}(\boldsymbol{u})\textnormal{d}\boldsymbol{u}\int_{[0,s\wedge t]}\frac{1}{\Bar{H}(u|\boldsymbol{z})\left(1-\Delta \Lambda(u|\boldsymbol{z})\right)}\textnormal{d}\Lambda(u|\boldsymbol{z}),
  \end{align*}
  where $\Bar{F}(t|\z)=1-F(t|\z)$ and $\Bar{H}(t|\z)=1-H(t-|\z)$.
\end{theorem}
\begin{corollary}[Weak convergence bias of the conditional expert Kaplan--Meier estimator]\label{bias_weakconvergence} Let Assumptions~\ref{weak_convergence_assumptions_2} hold and assume that conditional contaminated right-censoring is entirely random.
    Let $\gamma(\cdot;\z)$ be a bias function of the expert satisfying 
\begin{align}
\label{weak_convergence_errorterm}
\bigg \vert \sqrt{n|\boldsymbol{B}_{n}|}\left(\E\left[\md(t;\boldsymbol{z})\right]-\left(H_{1}(t|\z)g(\z)+\gamma(t;\z)\right)\right)\bigg \vert  \to 0 
\end{align}
for every $t\geq 0$. If $H_{1}(\cdot|\z)g(\z)+\gamma(\cdot;\z)$ is a non-decreasing and bounded càdlàg function it holds for $0\leq \theta < H^{-1}(1|\z)$ that
\begin{align*}
        &\sqrt{n|\boldsymbol{B}_{n}|}\left(\LLd(\cdot |\boldsymbol{z})-\left(\Lambda(\cdot |\boldsymbol{z})+\Gamma(\cdot;\z)\right)\right)\overset{\textnormal{D}}{\to}\amsmathbb{L}(\cdot|\boldsymbol{z}),
        \\
        &\sqrt{n|\boldsymbol{B}_{n}|}\left(\phi\left(-\LLd(\cdot |\boldsymbol{z})\right)-\phi\left(-\Lambda(\cdot|\z)-\Gamma(\cdot;\z)\right)\right)\overset{\textnormal{D}}{\to}\amsmathbb{Z}(\cdot |\boldsymbol{z})
    \end{align*}
    in $\ell^{\infty}([0,\theta])$. The convergence can be guaranteed in finite dimensional distributions without any càdlàg requirement.
\end{corollary}
% We note that Corollary \ref{bias_weakconvergence} implies that for fixed every fixed vector $\boldsymbol{t}\in \amsmathbb{R}^{l}$ the estimators are asymptotically distributed as 
% \begin{align*}
%     \LLd(\boldsymbol{t} |\boldsymbol{z})&\overset{\text{as}}{\sim}\mathcal{N}\left(\Lambda(\boldsymbol{t}|\z)+\Gamma(\boldsymbol{t};\z),\Sigma_{\amsmathbb{L}(\cdot|\z)}(\boldsymbol{t})\right),
%     \\
%     \amsmathbb{F}_{\dagger}(\boldsymbol{t}|\z)&\overset{\text{as}}{\sim}\mathcal{N}\left(\phi(-\Lambda(\boldsymbol{t}|\z)-\Gamma(\boldsymbol{t};\z)),\Sigma_{\amsmathbb{Z}(\cdot|\z)}(\boldsymbol{t})\right).
% \end{align*}
\begin{example}[Example on biased judgments] 
Consider expert judgments as in Example~\ref{example_consistency}; that is, consider iid expert judgments $(\eta_{i})$  
\begin{align*}
    \eta_{i} \sim \delta_{i}\,\text{Bern}(\hat{\mathbb{p}}(W_{i},\boldsymbol{Z}_{i})),
\end{align*}
for $\hat{\mathbb{p}}(w,\z) = (1-p_{0}) + p_{0} p(w,\z)$. This again implies a bias function $\Gamma(\cdot;\z)$ given by 
\begin{align*}
    \Gamma(t;\z) = (1-p_{0}) \int_{[0,t]} \frac{1 - p(s,\z)}{1 - H(s-|\z)} \mathrm{d}H_{1}^{\times}(s|\z), \quad 0 \leq t < H^{-1}(1|\z),
\end{align*}
and the convergence holds in $\ell^{\infty}([0,\theta])$. For instance, this implies that 
\begin{align*}
    \amsmathbb{F}^{(n)}_{\dagger}(t|\z) \overset{\text{as}}{\sim} \mathcal{N}\left(\text{e}^{-\Gamma(t;\z)} F(t|\z), \sigma^{2}_{\amsmathbb{Z}(\cdot|\z)}(t,t)\right),
\end{align*}
under continuity assumptions on $\Lambda(\cdot|\z)$ and $\Gamma(\cdot;\z)$.
% , we then have the  it holds by Corollary \ref{bias_weakconvergence} that for every fixed vector $\boldsymbol{t}\in \amsmathbb{R}^{l}$ the expert estimators are asymptotically distributed as 
% \begin{align*}
%     \LLd(\boldsymbol{t} |\boldsymbol{z})&\overset{\text{as}}{\sim}\mathcal{N}\left(\Lambda(\boldsymbol{t}|\z)+\Gamma(\boldsymbol{t};\z),\Sigma_{\amsmathbb{L}(\cdot|\z)}(\boldsymbol{t})\right),
%     \\
%     \amsmathbb{F}_{\dagger}(\boldsymbol{t}|\z)&\overset{\text{as}}{\sim}\mathcal{N}\left(\text{e}^{-\Gamma(t;\z)}F(\boldsymbol{t}|\z),\Sigma_{\amsmathbb{Z}(\cdot|\z)}(\boldsymbol{t})\right).
% \end{align*}
\end{example}
\section{Numerical studies}\label{section_numerical}
This section presents numerical examples using both simulated and real-world data. The simulated data set concerns disability insurance claims, while the real-world data set focuses on bank loan defaults. We consider the conditional expert Kaplan–Meier estimator under both perfect and imperfect expert judgments. The simulation illustrates the finite-sample behavior of the estimator and its relation to theoretical asymptotic properties.

The bank loan data cover the period from June 2007 to December 2011, encompassing the financial crisis. For each loan, the data include the issue date, the next payment date for active loans, the last payment date for fully repaid loans, and the default date for loans with repayment failures. Covariates include the loan's interest rate and the borrower's debt-to-income ratio, which are commonly used in credit risk modelling due to their relevance to installment size and repayment capacity. We apply the conditional expert Kaplan–Meier estimator using these covariates to estimate time-to-default. Expert judgments are applied selectively to regions of the two-dimensional covariate space to reflect scenarios in which practitioners have contamination knowledge limited to specific covariate levels.

In practical applications, the choice of bandwidth is critical. We present a procedure for selecting a finite-sample optimal bandwidth for the functional estimator.

\textbf{Bandwidth selection.}
Bandwidth selection is studied, for instance, in~\cite{KS} and~\cite{Hardle}, as well as the references therein. There are two general approaches: either to select a bandwidth that is asymptotically optimal according to some criterion, or to select a bandwidth that is optimal for finite samples. In the latter case, the bandwidth is typically determined through least squares cross-validation. For this study, we consider the finite sample case and a form of least squares cross-validation motivated as follows. The usual least squares cross-validation for $\amsmathbb{H}^{(n)}$, for a fixed $t>0$, is given by
\begin{align*}
     \underset{\B}{\text{argmin }} \text{CV}\left(\amsmathbb{H}^{(n)}(t|\cdot), \B\right) 
     = \underset{\B}{\text{argmin}} \left\{\frac{1}{n}\sum_{m=1}^{n}\left(1_{(W_{nm}\leq t)}-\amsmathbb{H}_{-m}^{(n)}(t|\boldsymbol{Z}_{nm};\B)\right)^{2}\right\},
\end{align*}
where $\amsmathbb{H}_{-m}^{(n)}(t|\boldsymbol{Z}_{nm};\B)$ is the leave-one-out estimator defined as
\begin{align*}
     \amsmathbb{H}_{-m}^{(n)}(t|\boldsymbol{Z}_{nm};\B) = \sum_{j=1, j\neq m}^{n} \frac{1_{(W_{nj}\leq t)} K_{\B}\left(\Z_{nj}-\boldsymbol{Z}_{nm}\right)}{\sum_{l=1, l\neq m}^{n} K_{\B}\left(\Z_{nl}-\boldsymbol{Z}_{nm}\right)}.
\end{align*}

When viewing the conditional expert Kaplan--Meier estimator as a functional estimator, it is natural to ask for a bandwidth $\B$ that accounts for all $t$ simultaneously for both $\amsmathbb{H}^{(n)}$ and $\amsmathbb{H}_{\dagger,1}^{(n)}$. To this end, we consider a least squares cross-validation of the form
\begin{align*}
     \underset{\B}{\text{argmin}} \int_{[0,t]} w(s) \cdot \bigg\vert \bigg\vert 
     \left\{\text{CV}(\amsmathbb{H}^{(n)}(s|\cdot),\B), \text{CV}(\amsmathbb{H}_{\dagger,1}^{(n)}(s|\cdot), \B)\right\} 
     \bigg\vert \bigg\vert^{2} \mathrm{d}s,
\end{align*}
where we allow for a time-dependent weighting function $w$. This approach is similar to that in~\cite{ReferenceCurves}. The calculations required for the leave-one-out estimator can be computationally intensive, but a neat numerical implementation is available. In particular, it holds that
\begin{align*}
     &\text{CV}\left(\amsmathbb{H}^{(n)}(t|\cdot),\B\right) \\
     &\quad = \frac{1}{n} \sum_{m=1}^{n} \left(1_{(W_{nm}\leq t)} - \amsmathbb{H}^{(n)}(t|\boldsymbol{Z}_{nm};\B)\right)^{2} 
     \left(\frac{1_{(W_{nm}\leq t)} - \amsmathbb{H}_{-m}^{(n)}(t|\boldsymbol{Z}_{nm};\B)}{1_{(W_{nm}\leq t)} - \amsmathbb{H}^{(n)}(t|\boldsymbol{Z}_{nm};\B)}\right)^{2},
\end{align*}
and that
\begin{align*}
     \frac{1_{(W_{nm}\leq t)} - \amsmathbb{H}^{(n)}(t|\boldsymbol{Z}_{nm};\B)}{1_{(W_{nm}\leq t)} - \amsmathbb{H}^{(n)}_{-m}(t|\boldsymbol{Z}_{nm};\B)}
     = 1 - \frac{K_{\B}\left(\boldsymbol{0}\right)}{\sum_{j=1}^{n} K_{\B}\left(\boldsymbol{Z}_{nj} - \boldsymbol{Z}_{nm}\right)}.
\end{align*}

The $n(n-1)/2+1$ unique expressions of $K_{\B}(\Z_{nj}-\Z_{nm})$ can be stored and reused efficiently at each step of an optimization routine.

\textbf{Simulated data set.} We consider an insurance portfolio of $n=\text{10,000}$ disability covers, for which we are interested in estimating the disability rates. We model the disability intensity as dependent on the insured's age and use this as a covariate for estimation. A complicated form of contamination is simulated, where the baseline of contamination intensity depends on a number of reportings. We interpret reportings broadly, for instance as an internal reporting number set by the insurance company indicating that a policyholder may be prone to contamination. The censoring variables $(C_{i})$ are simulated as iid $\text{unif}([0,50])$ variables, and the iid covariates $(\boldsymbol{Z}_{i})=(Z_{\text{age},i},Z_{\text{reportings},i})$ are simulated as
\begin{align*}
    &Z_{\text{age},i}\sim \mathcal{N}\left(\mu=50, \sigma^{2}=100\right), \quad 
    Z_{\text{reportings},i}\sim \text{Pois}\left(\lambda = 0.3\right).
\end{align*}
Given the covariate samples, we simulate disability events $(X_{i}|Z_{\text{age},i})$ and contamination events $(Y_{i}|Z_{\text{reportings},i})$ according to the conditional intensities 
\begin{align*}
    \mu_{\text{disability}}(t|z_{\text{age}})&=0.01\cdot \exp\left(0.02\cdot (t+z_{\text{age}})\right),\\
    \mu_{\text{contamination}}(t|z_{\text{reportings}})&=0.005 + 0.02 \cdot z_{\text{reportings}}.
\end{align*}
Consequently, the samples on disability, contamination, and censoring are simulated as conditionally mutually independent. Note that the disability intensity is increasing in $z_{\text{age}}\in\amsmathbb{R}$, and the contamination intensity is increasing in $z_{\text{reportings}}\in\amsmathbb{N}_{0}$. 

We consider the following Gaussian kernel with compact support:
\begin{align*}
     K(z) &= 1_{[-2,2]}(z) \frac{\phi(z)}{\Phi(2)-\Phi(-2)},
\end{align*}
and for the first study, we consider an expert with 100\% effectiveness providing iid samples
\begin{align*}
\eta^{(1)}_{i} \sim \delta_{i}\, \text{Bern}(p(W_{i},\boldsymbol{Z}_{i})), %\label{expert100efficient}
\end{align*}
where $p$ mirrors that of Remark~\ref{bandwidth_remark_consistency}. In this setup, $p$ is given as the following fraction of densities:
\begin{align*}
     p(w,\boldsymbol{z}) = \frac{f_{\text{disability}}(w|z_{\text{age}})}{f_{\text{disability}}(w|z_{\text{age}}) + f_{\text{contamination}}(w|z_{\text{reportings}})}.
\end{align*}
We examine how the perfect expert estimator fits the true survival function of disability in comparison with the naïve conditional Kaplan--Meier estimator in the heat plot in Figure~\ref{heatplot_simulated1}.
\begin{figure}[h]
    \centering
    \captionsetup{width=12cm}
    \includegraphics[width=4.5cm]{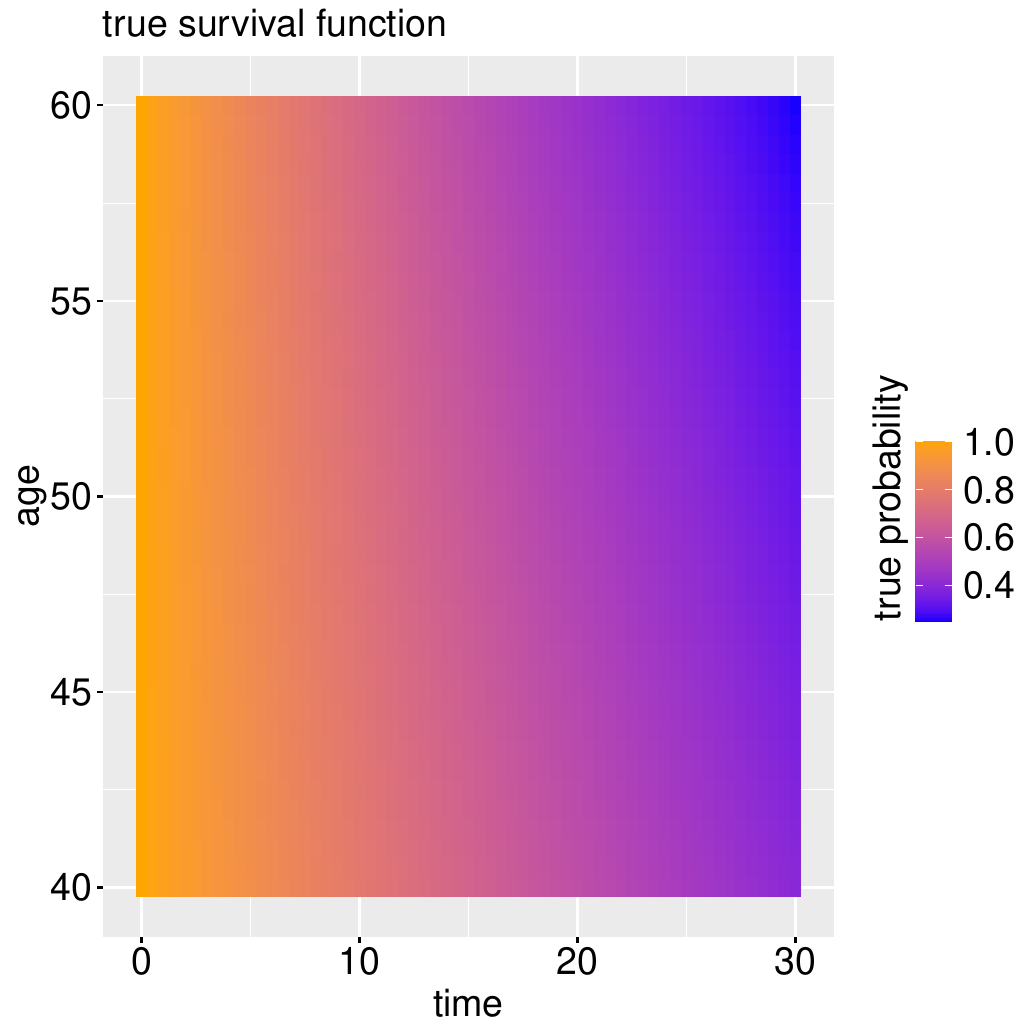}
    \includegraphics[width=4.5cm]{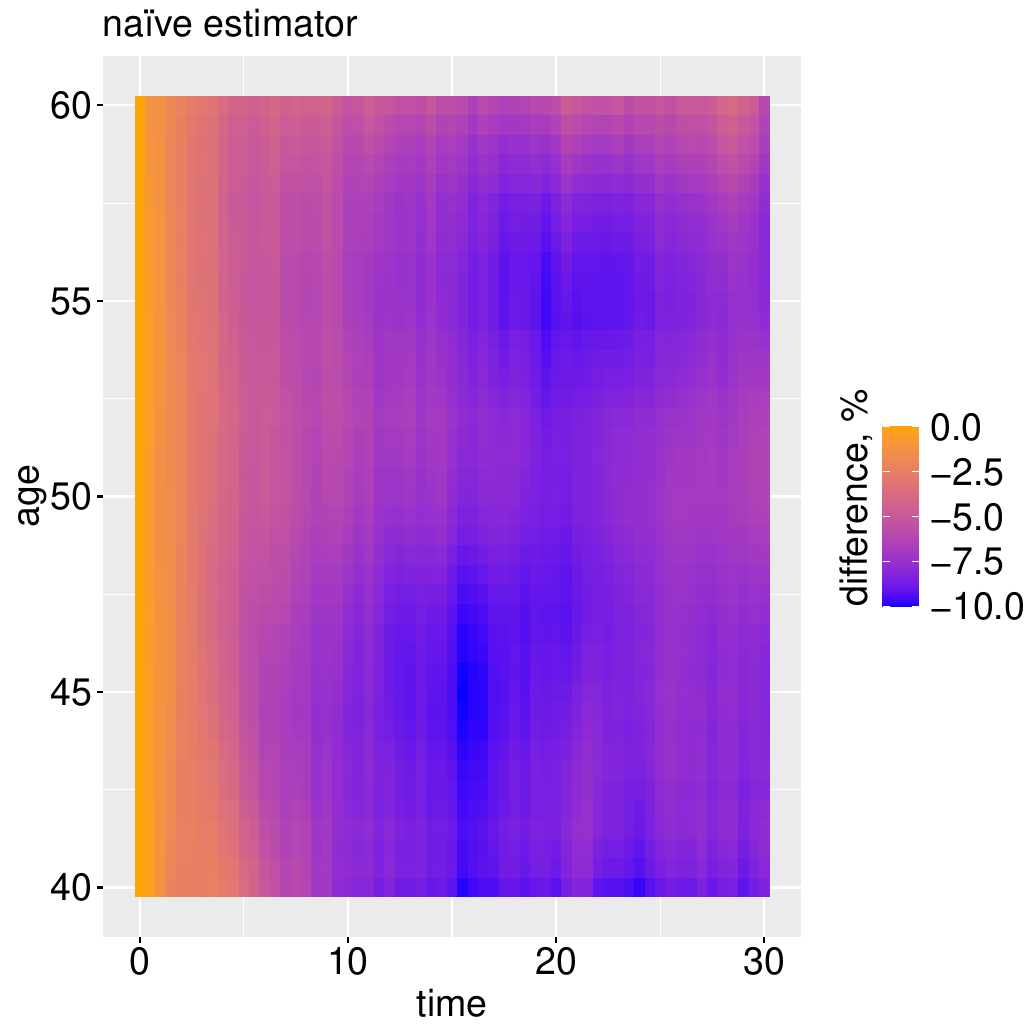}
    \includegraphics[width=4.5cm]{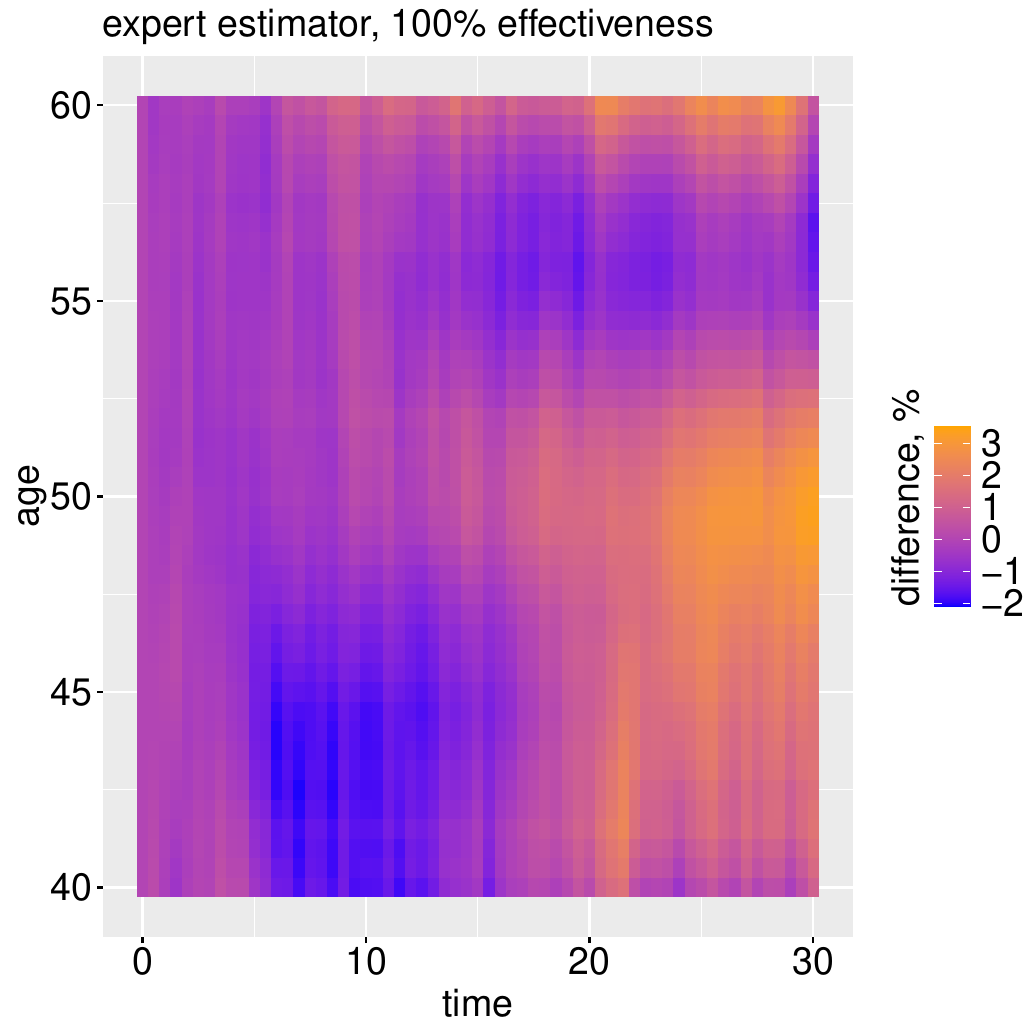}
    \caption{\textit{The true survival function $1-F$ (left); and the differences of the naïve estimator and  expert estimator of 100\% effectiveness (mid and right, respectively) to $1-F$ for $t\in[0,30]$ and $z_{\text{age}}\in[40,60]$.}}
    \label{heatplot_simulated1}
\end{figure}

Incorporating an expert with full information of the contamination mechanism in the conditional Kaplan--Meier estimator yields a better fit compared with the naïve estimator, as expected. This improvement is especially pronounced for longer times. The naïve estimator never overestimates the true survival function, which is unreasonable and indicates its bias. In contrast, the expert estimator both overestimates and underestimates the true survival function over an overall smaller interval than the naïve estimator, without any particular pattern.

To examine weak convergence, we simulate 300 data sets of 10,000 observations each according to the above specifications. For this study, we also introduce another expert with only 85\% effectiveness:
\begin{align*}
     \eta^{(0.85)}_{i} \sim \delta_{i} \,\text{Bern}\left((1-0.85) + 0.85 \cdot p(W_{i},\boldsymbol{Z}_{i})\right).
\end{align*}
According to Corollary~\ref{bias_weakconvergence}, it is anticipated that this estimator is normally distributed around its error function determined by $\eta^{(0.85)}$. We focus on the cases $z_{\text{age}} = 45, 50, 55$ for $t = 2, 5, 10$ to quantify, for each of these times, the convergence to the true mean and standard deviations of the underlying Gaussian processes for both the expert with 85\% effectiveness and the expert with 100\% effectiveness.

\begin{figure}[h]

    \centering
%\captionsetup{justification=centering}    
\includegraphics[,width=9cm]{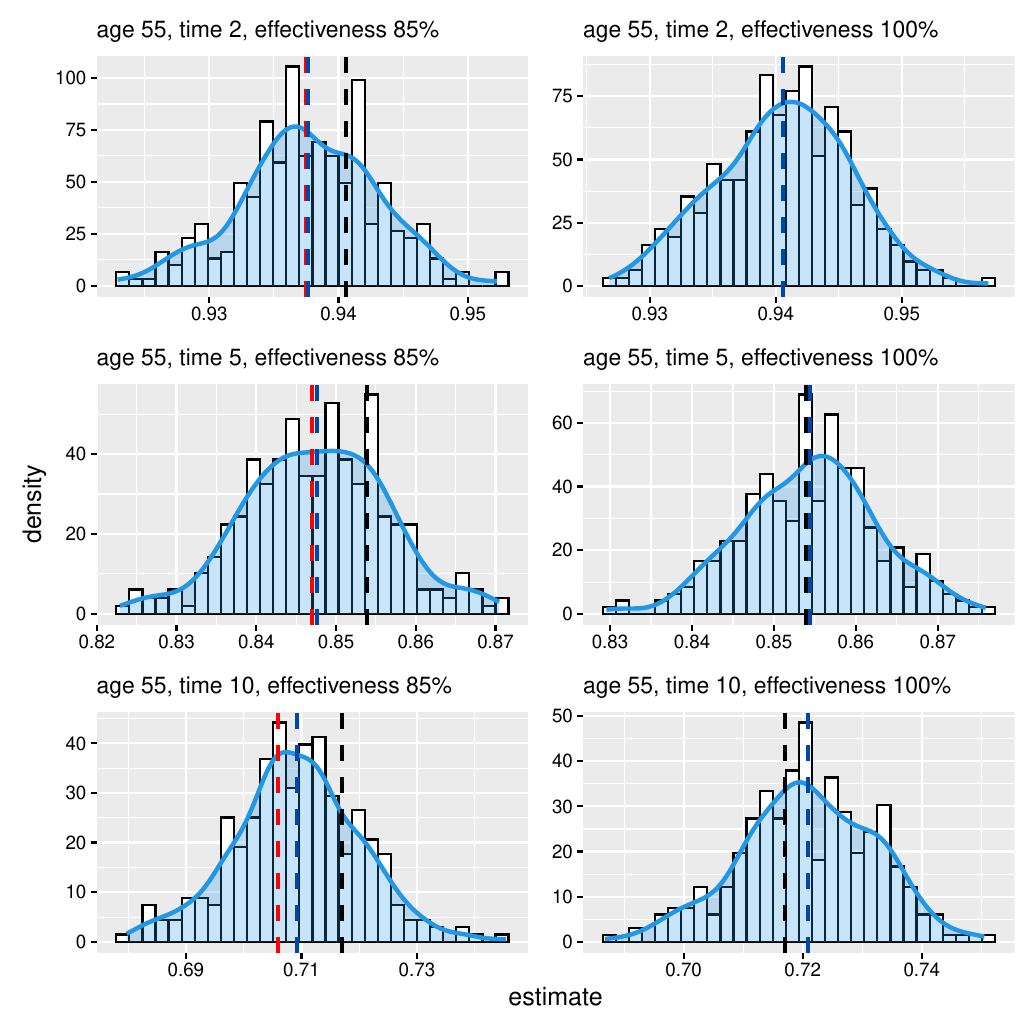} \\
      \scriptsize \trueunbiased \text{ } true unbiased mean, \truebiased \text{ }true biased mean, \estimate \text{ }estimated mean
      \captionsetup{width=12cm}
      \caption{\textit{The simulated survival distributions of the expert estimators of 85\% (left panel) and 100\% effectiveness (right panel) for $z_{\text{age}}=55$ and $t=2,5,10$.}}
          \label{simulated_normality_plot}
\end{figure}
\begin{table}
\captionsetup{width=12cm}
\caption{The simulated survival means, standard deviations and errors for the expert estimators of 85\% and 100\% effectiveness for $z_{\text{age}}=45,50,55$ and $t=2,5,10$.}
\resizebox{12cm}{!}{
\begin{tabular}{cclllllll}
\multicolumn{1}{l}{age}                   & \multicolumn{1}{l}{effect.}           & time                    & true (biased) mean          & simul. mean              & err. (mean)                 & true std. dev.               & simul. sd         & err. (sd)            \\ \hline
\multicolumn{1}{|c|}{\multirow{6}{*}{45}} & \multicolumn{1}{c|}{\multirow{3}{*}{85\%}}  & \multicolumn{1}{l|}{2}  & \multicolumn{1}{l|}{0.9480} & \multicolumn{1}{l|}{0.9483} & \multicolumn{1}{l|}{0.03\%}  & \multicolumn{1}{l|}{0.0065}  & \multicolumn{1}{l|}{0.0044} & \multicolumn{1}{l|}{-0.21\%} \\ \cline{3-9} 
\multicolumn{1}{|c|}{}                    & \multicolumn{1}{c|}{}                       & \multicolumn{1}{l|}{5}  & \multicolumn{1}{l|}{0.8716} & \multicolumn{1}{l|}{0.8718} & \multicolumn{1}{l|}{0.02\%}  & \multicolumn{1}{l|}{0.0101}  & \multicolumn{1}{l|}{0.0059} & \multicolumn{1}{l|}{-0.42\%} \\ \cline{3-9} 
\multicolumn{1}{|c|}{}                    & \multicolumn{1}{c|}{}                       & \multicolumn{1}{l|}{10} & \multicolumn{1}{l|}{0.7498} & \multicolumn{1}{l|}{0.7510} & \multicolumn{1}{l|}{0.12\%}  & \multicolumn{1}{l|}{0.0138} & \multicolumn{1}{l|}{0.0085} & \multicolumn{1}{l|}{-0.53\%} \\ \cline{2-9} 
\multicolumn{1}{|c|}{}                    & \multicolumn{1}{c|}{\multirow{3}{*}{100\%}} & \multicolumn{1}{l|}{2}  & \multicolumn{1}{l|}{0.9510} & \multicolumn{1}{l|}{0.9512} & \multicolumn{1}{l|}{0.02\%}  & \multicolumn{1}{l|}{0.0065}  & \multicolumn{1}{l|}{0.0043} & \multicolumn{1}{l|}{-0.22\%} \\ \cline{3-9} 
\multicolumn{1}{|c|}{}                    & \multicolumn{1}{c|}{}                       & \multicolumn{1}{l|}{5}  & \multicolumn{1}{l|}{0.8787} & \multicolumn{1}{l|}{0.8784} & \multicolumn{1}{l|}{-0.03\%} & \multicolumn{1}{l|}{0.0101}  & \multicolumn{1}{l|}{0.0059} & \multicolumn{1}{l|}{-0.42\%} \\ \cline{3-9} 
\multicolumn{1}{|c|}{}                    & \multicolumn{1}{c|}{}                       & \multicolumn{1}{l|}{10} & \multicolumn{1}{l|}{0.7616} & \multicolumn{1}{l|}{0.7630} & \multicolumn{1}{l|}{0.14\%}  & \multicolumn{1}{l|}{0.0065}  & \multicolumn{1}{l|}{0.0082} & \multicolumn{1}{l|}{-0.17\%} \\ \hline
\multicolumn{1}{|l|}{\multirow{6}{*}{50}} & \multicolumn{1}{c|}{\multirow{3}{*}{85\%}}  & \multicolumn{1}{l|}{2}  & \multicolumn{1}{l|}{0.9430} & \multicolumn{1}{l|}{0.9435} & \multicolumn{1}{l|}{0.05\%}  & \multicolumn{1}{l|}{0.0064}  & \multicolumn{1}{l|}{0.0043} & \multicolumn{1}{l|}{-0.21\%} \\ \cline{3-9} 
\multicolumn{1}{|l|}{}                    & \multicolumn{1}{c|}{}                       & \multicolumn{1}{l|}{5}  & \multicolumn{1}{l|}{0.8599} & \multicolumn{1}{l|}{0.8609} & \multicolumn{1}{l|}{0.10\%}  & \multicolumn{1}{l|}{0.0098}  & \multicolumn{1}{l|}{0.0066} & \multicolumn{1}{l|}{-0.32\%} \\ \cline{3-9} 
\multicolumn{1}{|l|}{}                    & \multicolumn{1}{c|}{}                       & \multicolumn{1}{l|}{10} & \multicolumn{1}{l|}{0.7284} & \multicolumn{1}{l|}{0.7315} & \multicolumn{1}{l|}{0.31\%}  & \multicolumn{1}{l|}{0.0133}  & \multicolumn{1}{l|}{0.0089} & \multicolumn{1}{l|}{-0.44\%} \\ \cline{2-9} 
\multicolumn{1}{|l|}{}                    & \multicolumn{1}{c|}{\multirow{3}{*}{100\%}} & \multicolumn{1}{l|}{2}  & \multicolumn{1}{l|}{0.9460} & \multicolumn{1}{l|}{0.9464} & \multicolumn{1}{l|}{0.04\%}  & \multicolumn{1}{l|}{0.0064}  & \multicolumn{1}{l|}{0.0041} & \multicolumn{1}{l|}{-0.23\%} \\ \cline{3-9} 
\multicolumn{1}{|l|}{}                    & \multicolumn{1}{c|}{}                       & \multicolumn{1}{l|}{5}  & \multicolumn{1}{l|}{0.8668} & \multicolumn{1}{l|}{0.8679} & \multicolumn{1}{l|}{0.11\%}  & \multicolumn{1}{l|}{0.0098}  & \multicolumn{1}{l|}{0.0062} & \multicolumn{1}{l|}{-0.36\%} \\ \cline{3-9} 
\multicolumn{1}{|l|}{}                    & \multicolumn{1}{c|}{}                       & \multicolumn{1}{l|}{10} & \multicolumn{1}{l|}{0.7401} & \multicolumn{1}{l|}{0.7438} & \multicolumn{1}{l|}{0.37\%}  & \multicolumn{1}{l|}{0.0133}  & \multicolumn{1}{l|}{0.0085} & \multicolumn{1}{l|}{-0.48\%} \\ \hline
\multicolumn{1}{|c|}{\multirow{6}{*}{55}} & \multicolumn{1}{c|}{\multirow{3}{*}{85\%}}  & \multicolumn{1}{l|}{2}  & \multicolumn{1}{l|}{0.9375} & \multicolumn{1}{l|}{0.9377} & \multicolumn{1}{l|}{0.02\%}  & \multicolumn{1}{l|}{0.0071}  & \multicolumn{1}{l|}{0.0053} & \multicolumn{1}{l|}{-0.18\%} \\ \cline{3-9} 
\multicolumn{1}{|c|}{}                    & \multicolumn{1}{c|}{}                       & \multicolumn{1}{l|}{5}  & \multicolumn{1}{l|}{0.8470} & \multicolumn{1}{l|}{0.8476} & \multicolumn{1}{l|}{0.06\%}  & \multicolumn{1}{l|}{0.0109}  & \multicolumn{1}{l|}{0.0089} & \multicolumn{1}{l|}{-0.20\%} \\ \cline{3-9} 
\multicolumn{1}{|c|}{}                    & \multicolumn{1}{c|}{}                       & \multicolumn{1}{l|}{10} & \multicolumn{1}{l|}{0.7059} & \multicolumn{1}{l|}{0.7092} & \multicolumn{1}{l|}{0.33\%} & \multicolumn{1}{l|}{0.0146}  & \multicolumn{1}{l|}{0.0109} & \multicolumn{1}{l|}{-0.37\%} \\ \cline{2-9} 
\multicolumn{1}{|c|}{}                    & \multicolumn{1}{c|}{\multirow{3}{*}{100\%}} & \multicolumn{1}{l|}{2}  & \multicolumn{1}{l|}{0.9405} & \multicolumn{1}{l|}{0.9406} & \multicolumn{1}{l|}{0.01\%}  & \multicolumn{1}{l|}{0.0071}  & \multicolumn{1}{l|}{0.0053} & \multicolumn{1}{l|}{-0.18\%} \\ \cline{3-9} 
\multicolumn{1}{|c|}{}                    & \multicolumn{1}{c|}{}                       & \multicolumn{1}{l|}{5}  & \multicolumn{1}{l|}{0.8539} & \multicolumn{1}{l|}{0.8544} & \multicolumn{1}{l|}{0.05\%}  & \multicolumn{1}{l|}{0.0109}  & \multicolumn{1}{l|}{0.0082} & \multicolumn{1}{l|}{-0.27\%} \\ \cline{3-9} 
\multicolumn{1}{|c|}{}                    & \multicolumn{1}{c|}{}                       & \multicolumn{1}{l|}{10} & \multicolumn{1}{l|}{0.7171} & \multicolumn{1}{l|}{0.7208} & \multicolumn{1}{l|}{0.37\%}  & \multicolumn{1}{l|}{0.0146}  & \multicolumn{1}{l|}{0.0113} & \multicolumn{1}{l|}{-0.33\%} \\ \hline
\end{tabular}}
\label{gaussian_table}
\end{table}

Table \ref{gaussian_table} reports the mean values and standard deviations of the simulated survival estimators. The results show that estimates for times $t=2$ and $t=5$ are closer to their true values than those at $t=10$, with similar performance across the two estimators. The relative error in standard deviations is larger than that for the means, likely reflecting the finite-sample variance dependence on the chosen bandwidth, whereas the mean is less sensitive due to symmetry in the estimates. Figure~\ref{simulated_normality_plot} presents the estimates for $z_{\text{age}} = 55$ as empirical densities on a probability scale. The figure illustrates how expert $\eta^{(1)}$ shifts the estimated survival curve closer to the true underlying curve compared to expert $\eta^{(0.85)}$.
\newpage
\textbf{Real world data set.} The bank loan default data contains two continuous covariates: interest rate (IR) and debt-to-income (DtI) as described above. We fit expert conditional Kaplan--Meier estimators using a two-dimensional regression on these covariates and model interest rates in the range $6\%$ to $12\%$ and debt-to-income ratios in the range $8\%$ to $20\%$. The data set consists of 10,130 rows. Figure~\ref{defaultpercentage_numberofobservations} shows that the data set contains a relatively high proportion of defaults, with default rates increasing for higher debt-to-income ratios and higher interest rates.
\begin{figure}[h]
\centering
\begin{tabular}{lrr}
                                 & \multicolumn{1}{c}{IR 6-9\%}  & \multicolumn{1}{c}{IR 9-12\%} \\ \cline{2-3} 
\multicolumn{1}{l|}{DtI 8-14\%}  & \multicolumn{1}{c|}{6.1\% \textbackslash \text{ }2,313}   & \multicolumn{1}{c|}{11.5\% \textbackslash \text{ }3,022}    \\ \cline{2-3} 
\multicolumn{1}{c|}{DtI 14-20\%} & \multicolumn{1}{c|}{7.0\% \textbackslash\text{ }1,832}  & \multicolumn{1}{c|}{13.5\% \textbackslash\text{ }2,963}   \\ \cline{2-3}
                                 & \multicolumn{2}{c}{(default \% \textbackslash\text{ no. of observations})}                    
\end{tabular}

% \begin{tabular}{lrr}
%                                  & IR 6-9\%                   & IR 9-12\%                  \\ \cline{2-3} 
% \multicolumn{1}{l|}{DtI 8-14\%}  & \multicolumn{1}{r|}{2,313} & \multicolumn{1}{r|}{3,022} \\ \cline{2-3} 
% \multicolumn{1}{l|}{DtI 14-20\%} & \multicolumn{1}{r|}{1,832} & \multicolumn{1}{r|}{2,963} \\ \cline{2-3} 
%                                  & \multicolumn{2}{c}{(number of observations)}           
% \end{tabular}
\captionsetup{width=0.85\textwidth}
\caption{\textit{Default percentage and number of observations across different covariate intervals for the loan default data.}}
\label{defaultpercentage_numberofobservations}
\end{figure}
\\
Boxplots in Figure~\ref{boxplots_defaultlength} show that, for defaulted loans, the number of months until default is similar across the data set. We examine how different expert information can be used to alter the estimation of survival probabilities across the two-dimensional covariate plane.  

We define two experts: one that assumes 5\% of the observations are contaminated uniformly across the plane, and another that assumes 5\% contamination only for observations with interest rates below $10\%$. Specifically, for iid $B_{i}\sim \text{Bern}(0.95)$, the experts are defined as
\begin{align*}
    \eta^{(\text{uniform})}_{i} = \delta_{i} \cdot B_{i},\quad
    \eta^{(\text{specific})}_{i} = \delta_{i} \cdot \left\{ 1_{(z_{\text{IR},i}>10\%)} + B_{i} \cdot 1_{(z_{\text{IR},i}\leq 10\%)}\right\}.
\end{align*}
A feature of interest for the specific expert is its ability to capture the shift in belief at $z_{\text{IR}}=10\%$ in its estimates. The specific expert estimator is expected to interpolate between the uniform expert estimator and the naïve estimator in a neighborhood around this threshold, depending on the optimal bandwidth selected via functional least-squares cross-validation.

\begin{figure}[h]
    \centering
    \includegraphics[width=11cm]{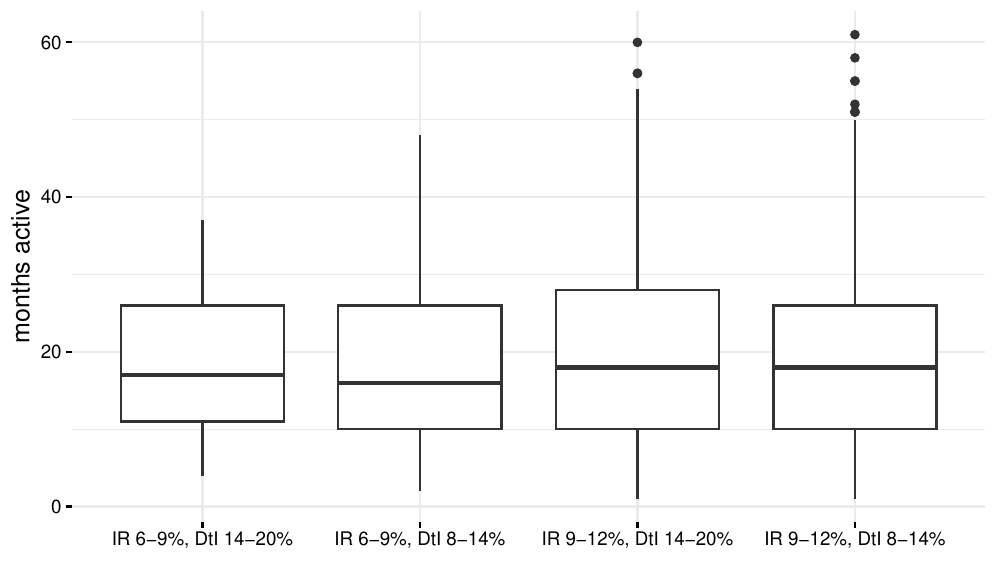}
    \captionsetup{width=12cm}
    \caption{\textit{Boxplots of months active for defaulted loans across different covariate levels for the loan default data.}}
    \label{boxplots_defaultlength}
\end{figure}

For the study, we choose a diagonal bandwidth matrix and a two-dimensional truncated Gaussian kernel with independent entries:
\begin{align*}
    \boldsymbol{B} = \begin{bmatrix}
        b_{1} & 0 \\
        0 & b_{2}
    \end{bmatrix},\quad
    K(z_{1},z_{2}) = \begin{bmatrix}
         1_{[-2,2]}(z_{1}) \frac{\phi(z_{1})}{\Phi(2)-\Phi(-2)} \\
         1_{[-2,2]}(z_{2}) \frac{\phi(z_{2})}{\Phi(2)-\Phi(-2)}
    \end{bmatrix}.
\end{align*}
Diagonal bandwidths are practical, as optimization over general bandwidths is computationally intensive. A general bandwidth choice can also be interpreted as standardizing the covariate data so that its covariance matrix is the identity.  

Applying these expert judgments to the data, the uniform expert lowers the default percentage across the entire data set, whereas the specific expert reduces the default percentage only for observations with interest rates below $10\%$, see Figure \ref{defaultpercentage_numberofobservations_experts}.

\begin{figure}[h]
\centering
\begin{tabular}{lrr}
                                 & \multicolumn{1}{c}{IR 6-9\%}  & \multicolumn{1}{c}{IR 9-12\%} \\ \cline{2-3} 
\multicolumn{1}{l|}{DtI 8-14\%}  & \multicolumn{1}{c|}{5.8\% \textbackslash\text{ }5.8\%}   & \multicolumn{1}{c|}{10.9\% \textbackslash\text{ }11.4\%}    \\ \cline{2-3} 
\multicolumn{1}{c|}{DtI 14-20\%} & \multicolumn{1}{c|}{6.7\% \textbackslash\text{ }6.7\%}  & \multicolumn{1}{c|}{12.7\% \textbackslash\text{ }13.3\%}   \\ \cline{2-3}
                                 & \multicolumn{2}{c}{(uniform expert \textbackslash\text{ specific expert})}                    
\end{tabular}
\captionsetup{width=0.85\textwidth}
\caption{\textit{Default percentages for expert-defined contamination scenarios on the loan default data.}}
\label{defaultpercentage_numberofobservations_experts}
\end{figure}

To measure how the experts modify the survival estimates compared to the naïve conditional Kaplan--Meier estimator $\amsmathbb{F}^{(n)}$, we compute the integral difference across the covariate plane (as presented in Figure~\ref{integral_differences_defaultdata}):
\begin{align*}
    \z = (z_{\text{DtI}}, z_{\text{IR}}) \mapsto \int_{[0,50]} (1-\amsmathbb{F}_{\dagger}^{(n)}(t|\z)) - (1-\amsmathbb{F}^{(n)}(t|\z)) \, \text{d}t.
\end{align*}

\begin{figure}[h]
    \centering
    \captionsetup{width=12cm}
    \includegraphics[width=6.5cm, height=6.5cm]{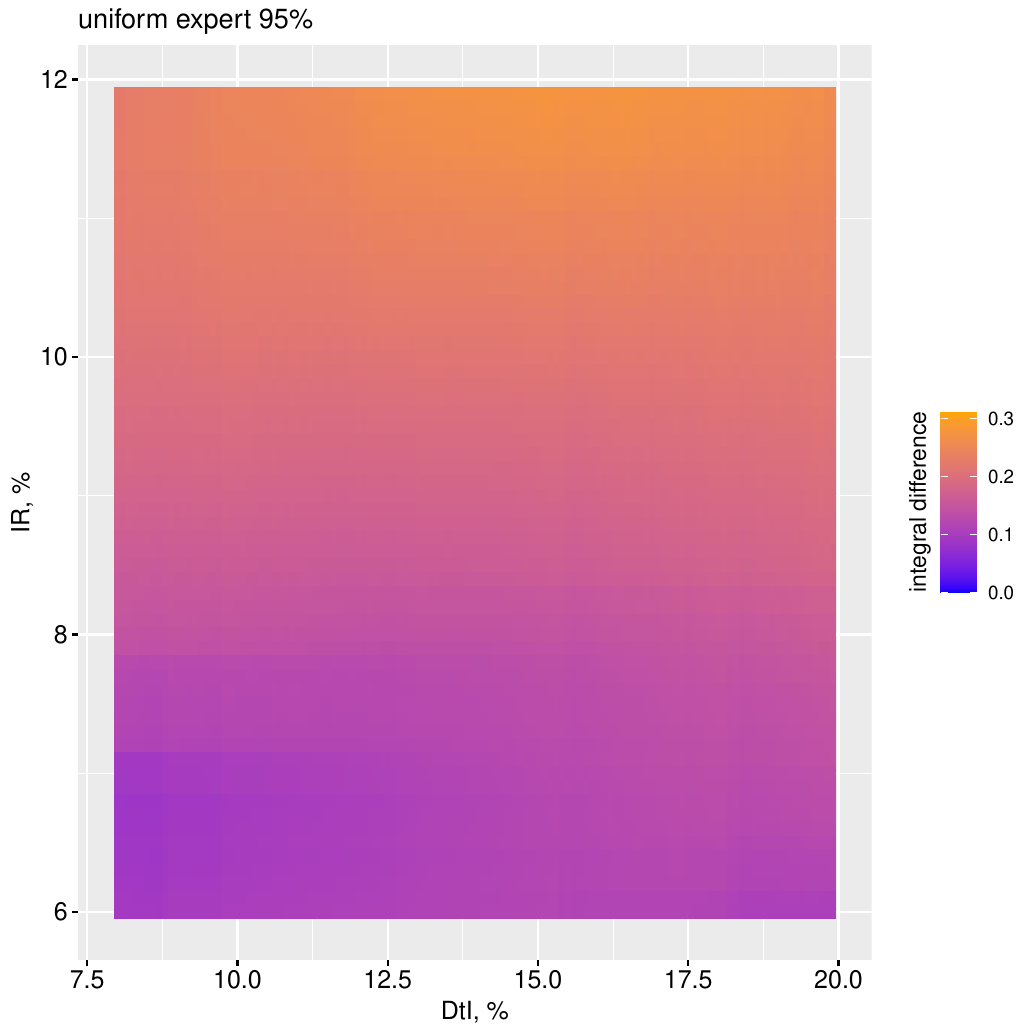}
    \includegraphics[width=6.5cm, height=6.5cm]{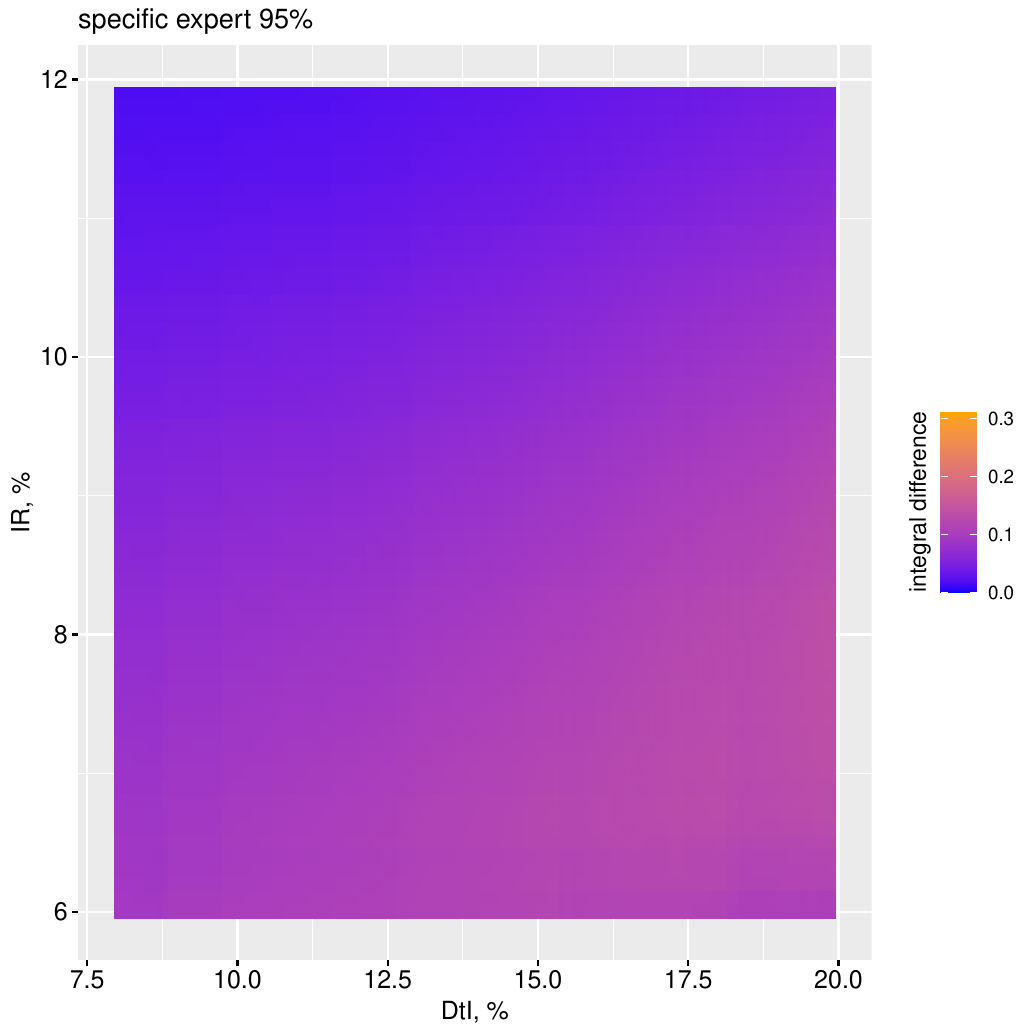}
    \caption{\textit{Integral differences of survival estimates between the expert estimators, uniform expert (left) and specific expert (right), and the naïve estimator on the covariate plane.}}
    \label{integral_differences_defaultdata}
\end{figure}

For interest rates above $10\%$, the specific expert estimator closely matches the naïve estimator, while for interest rates around and below $8\%$, it aligns with the uniform expert estimator. Figure~\ref{interestrate10_defaultdata} shows that at $z_{\text{IR}}=10\%$, the specific expert estimator is closer to the naïve estimator than the uniform expert estimator, though not exactly equal. In Figure~\ref{survivalplots_defaultdata}, this difference is negligible for $z_{\text{IR}}=11\%$.

\begin{figure}[h]
    \centering
    \includegraphics[width=4.5cm]{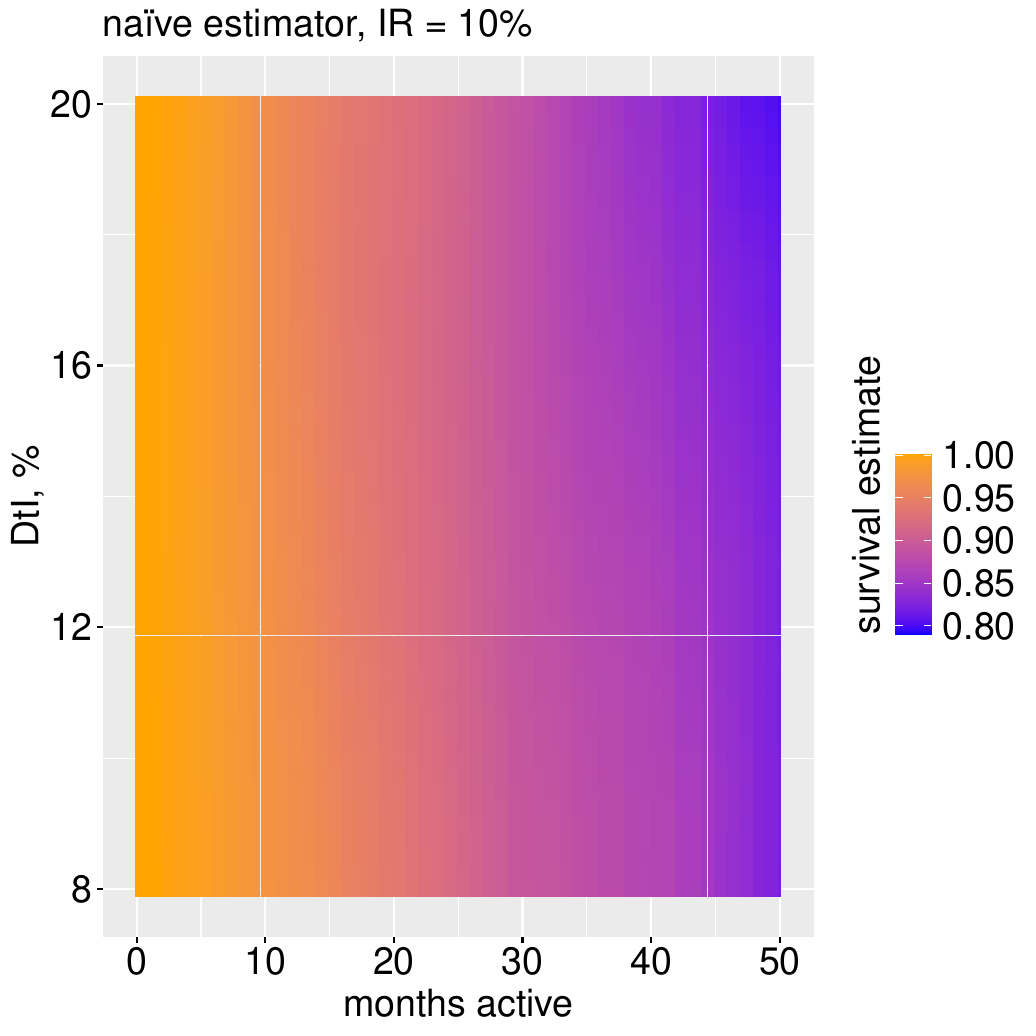}
    \includegraphics[width=4.5cm]{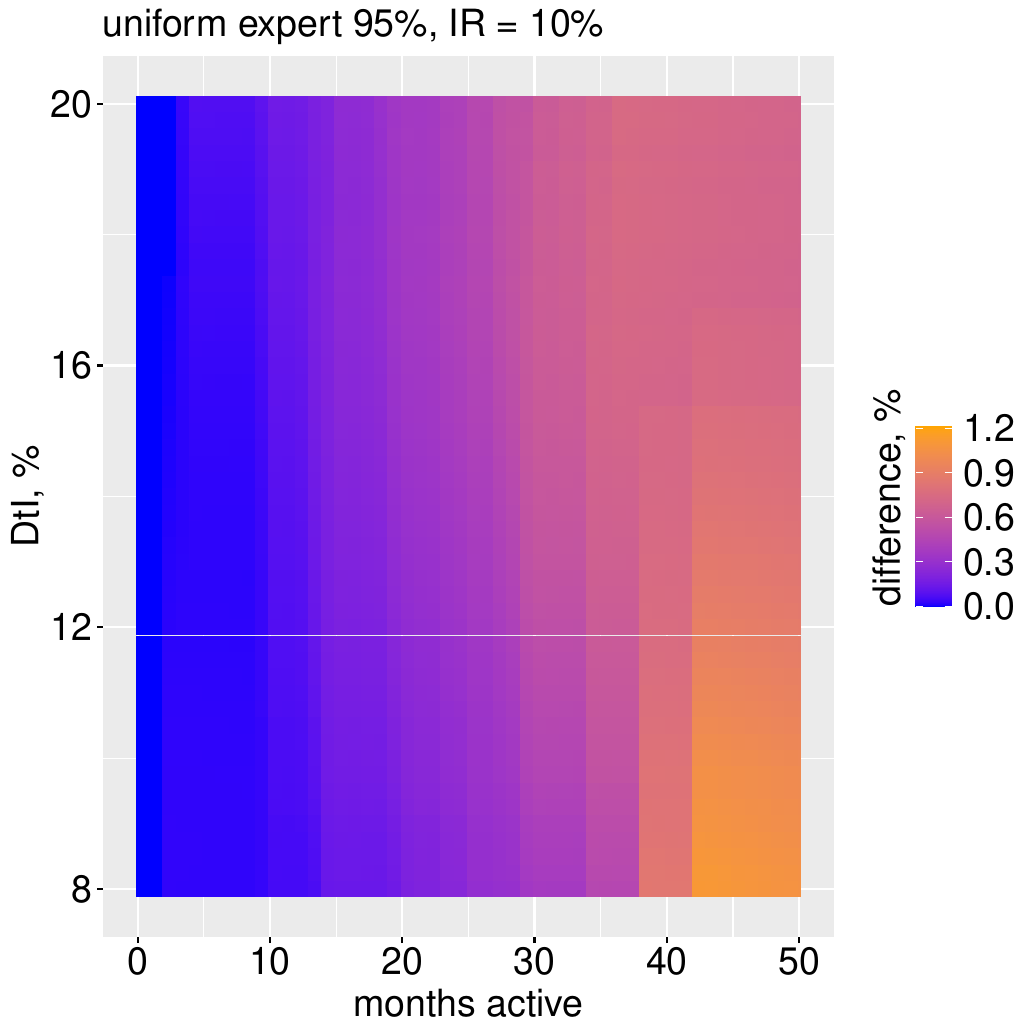}
    \includegraphics[width=4.5cm]{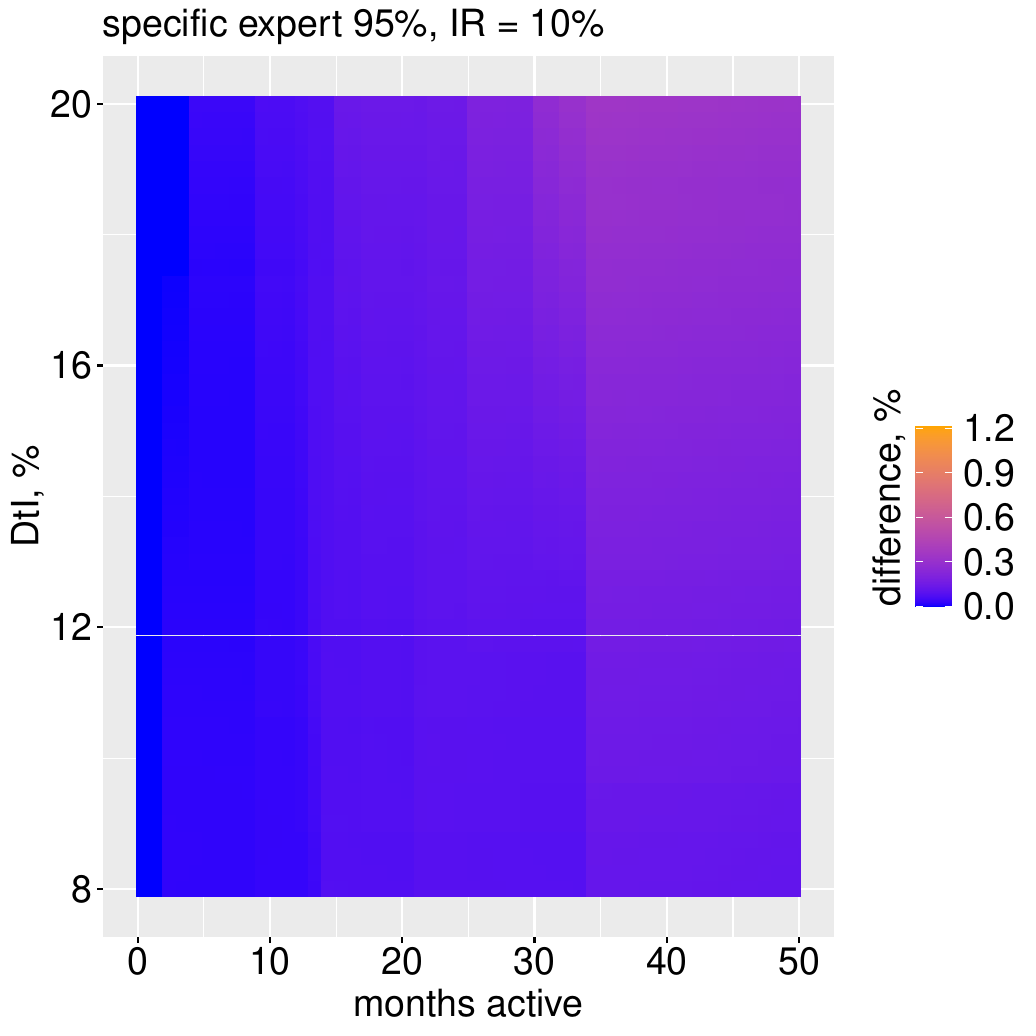}
    \captionsetup{width=12cm}
    \caption{\textit{Differences in survival estimates between the expert estimators (mid and right) and the naïve estimator for $z_{\text{IR}}=10\%$, shown to the left.}}
    \label{interestrate10_defaultdata}
\end{figure}

\begin{figure}[h]
    \centering
    \captionsetup{width=12cm}
    %\captionsetup{justification=centering} 
    \includegraphics[scale=0.6]{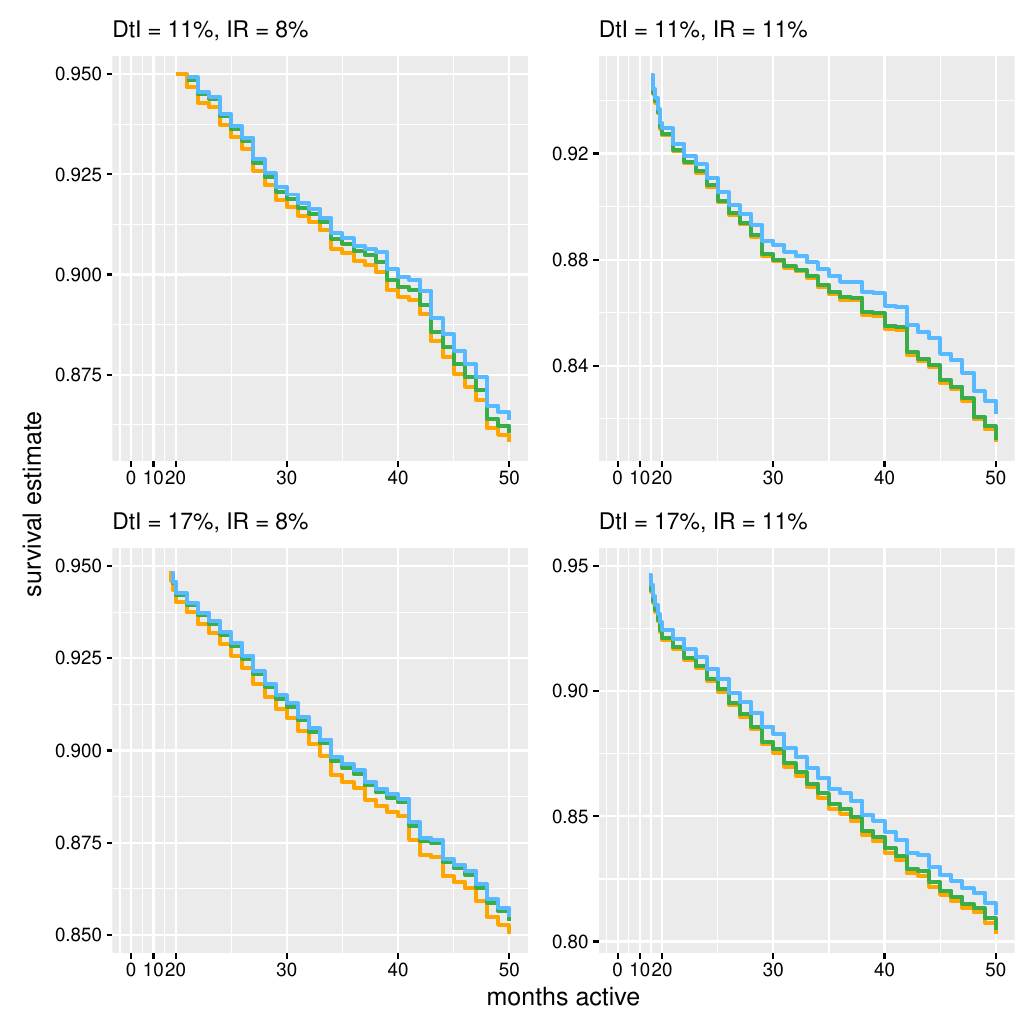}
    \\
    \scriptsize \naive \text{ } naïve, \uniform \text{ } uniform expert, \specific \text{ } specific expert
    \caption{\textit{Conditional Kaplan--Meier estimators for selected covariate values.}}
    \label{survivalplots_defaultdata}
\end{figure}
\section{Conclusion}\label{section_conclusion}
In this paper, we have introduced a flexible, non-parametric framework for estimating survival functions from right-censored data in the presence of contamination, where observed events may not reflect true outcomes. Building on the conditional expert Kaplan--Meier estimator, we have developed a comprehensive asymptotic theory incorporating covariates via kernel smoothing. Our results demonstrate that consistency of the estimator can be achieved when expert judgments are unbiased, while systematic deviations in expert information produce a deterministic, quantifiable asymptotic bias in the estimated survival functions. This provides a practical and interpretable tool to assess the reliability of survival estimates in settings where data contamination is unavoidable, such as insurance portfolios or credit risk applications.

Through simulation studies, we have illustrated both the behavior of the estimator under varying levels of expert effectiveness and the weak convergence of the resulting processes. In a real-world application to bank loan defaults, we have shown how incorporating expert knowledge can meaningfully adjust survival estimates, capturing both uniform and covariate-specific contamination effects. The approach allows practitioners to benchmark the influence of expert information and quantify its impact on inference, providing a systematic framework to improve estimates in contaminated settings.

Future work may explore extensions in several directions. One natural extension is to incorporate more structured expert knowledge in a Bayesian framework, such as neutral-to-the-right models, where beliefs are formalized as prior distributions, with~\cite{BladtGonzalezBayesian} providing a good starting point. Another promising avenue is the treatment of truncated data, for example through adaptations of the Lynden--Bell estimator~\cite{LyndelBell} and its time-reversed variants~\cite{Lagakos}, enabling the framework to address left- or right-truncation in addition to contamination. More broadly, the methodology offers a principled way to integrate human judgment into survival analysis, opening opportunities for robust and interpretable estimation in other applied settings.
% \begin{table}[h]
%  \centering
%  \begin{tabular}{c|l|r|r}
%  % & \multicolumn{1}{c|}{Text} & \multicolumn{1}{c|}{Text} & \multicolumn{1}{c|}{Text} & \multicolumn{1}{c|}{Text} & \multicolumn{1}{c|}{text}
%  \\
%  \hline
%  \parbox[t]{2mm}{\multirow{5}{*}{\rotatebox[origin=c]{90}{Interest rate}}} & text &&\\
%  & text &&\\
%  & text &&\\
%   & text &&\\
%  & text &&\\
%  \hline
%  \end{tabular}
%  \end{table}

%%%%%%%%%%%%%%%%%%%%%%%%%%%%%%%%%%%%%%%%%%%%%%
%% Example with single Appendix:            %%
%%%%%%%%%%%%%%%%%%%%%%%%%%%%%%%%%%%%%%%%%%%%%%
% \begin{appendix}
% \section*{Title}\label{appn} %% if no title is needed, leave empty \section*{}.
% Appendices should be provided in \verb|{appendix}| environment,
% before Acknowledgements.

% If there is only one appendix,
% then please refer to it in text as \ldots\ in the \hyperref[appn]{Appendix}.
% \end{appendix}
%%%%%%%%%%%%%%%%%%%%%%%%%%%%%%%%%%%%%%%%%%%%%%
%% Example with multiple Appendixes:        %%
%%%%%%%%%%%%%%%%%%%%%%%%%%%%%%%%%%%%%%%%%%%%%%
\begin{appendix}
\section{Functional consistency}\label{appA}
In this section we provide proofs of consistency Theorem~\ref{consistency_theorem} and its  bias Corollary~\ref{bias} for the conditional expert Kaplan--Meier estimator and the lemmas involved. Lemma~\ref{cadlag_bounded_as} expands pointwise almost sure convergence into functional convergence for certain functions, Lemma~\ref{Hoeffding} states Hoeffding's inequality and Lemma~\ref{weak_consistency_proof} provides weak consistency for regression estimators. Lastly, Lemma~\ref{consistency_lemma} gives consistency of the Nadaraya--Watson estimator with expert judgments.
\begin{lemma}[Expanding pointwise convergence] \label{cadlag_bounded_as}
    Assume that $T\colon \amsmathbb{R}\to \amsmathbb{R}$ is a non-decreasing and bounded càdlàg function. Assume that $(T^{(n)})$ is a sequence of random non-decreasing càdlàg maps $T^{(n)}\colon \amsmathbb{R}\to \amsmathbb{R}$ such that $T^{(n)}(x)\overset{\textnormal{a.s.}}{\to}T(x)$ and $T^{(n)}(x-)\overset{\textnormal{a.s.}}{\to}T(x-)$ for all $x\in \amsmathbb{R}$. It then holds that 
    \begin{align*}
        \sup_{x\in \amsmathbb{R}}|T^{(n)}(x)-T(x)|\overset{\textnormal{a.s.}}{\to}0.
    \end{align*}
\end{lemma}
\begin{proof}
    The proof is virtually identical to the proof of the ordinary Glivenko--Cantelli theorem. Given $\varepsilon>0$ choose some $k\in \amsmathbb{N}$ and a sequence 
    \begin{align*}
        -\infty=x_{0}<\ldots <x_{k}=\infty
    \end{align*}
    satisfying 
    \begin{align*}
        T(x_{j}-)-T(x_{j-1})<\varepsilon.
    \end{align*}
    By assumption it holds for each point in the partition $x_{j}$ that
    \begin{align*}
        |T^{(n)}(x_{j})-T(x_{j})|\overset{\text{a.s.}}{\to}0\text{ and } |T^{(n)}(x_{j}-)-T(x_{j}-)|\overset{\text{a.s.}}{\to}0.
    \end{align*}
    Let 
    \begin{align*}
        \Delta^{(n)}=\max_{j=1,...,k-1}\left\{|T^{(n)}(x_{j})-T(x_{j})|,|T^{(n)}(x_{j}-)-T(x_{j}-)|\right\}
    \end{align*}
     and let $x\in \amsmathbb{R}$ be arbitrary. By construction there is a $j$ such that $x\in [x_{j-1}, x_{j})$ and we have 
    \begin{align*}
    &T^{(n)}(x_{j-1})-T(x_{j-1})-\varepsilon \leq T^{(n)}(x_{j-1})-T(x_{j}-)\leq  T^{(n)}(x)-T(x)     
    % \\
    %     &T^{(n)}(x)-T(x)\leq T^{(n)}(x_{j}-)-T(x_{j-1})\leq T^{(n)}(x_{j}-)-T(x_{j}-)+\epsilon 
    %     \\
    %     % T^{(n)}(x)-T(x)&\geq T^{(n)}(x_{j-1})-T(x_{j}-)\geq T^{(n)}(x_{j-1})-T(x_{j-1})-\epsilon 
    %     % \\
    %     % &\implies 
    %     % T^{(n)}(x_{j-1})-T(x_{j-1})-\epsilon \leq T^{(n)}.
    \end{align*}
    and also
    \begin{align*}
        T^{(n)}(x)-T(x)\leq T^{(n)}(x_{j}-)-T(x_{j-1})\leq T^{(n)}(x_{j}-)-T(x_{j}-)+\varepsilon.
    \end{align*}
    Then 
    \begin{align*}
        &T^{(n)}(x_{j-1})-T(x_{j-1})-\varepsilon\leq T^{(n)}(x)-T(x)\leq T^{(n)}(x_{j}-)-T(x_{j}-)+\varepsilon
        \\
        &\implies 
        |T^{(n)}(x)-T(x)|\leq \Delta^{(n)}+\varepsilon.
    \end{align*}
    As $x$ was arbitrary it holds that 
    \begin{align*}
        \sup_{x\in \amsmathbb{R}}|T^{(n)}(x)-T(x)|\leq \Delta^{(n)}+\varepsilon.
    \end{align*}
    Note that $\Delta^{(n)}$ is the maximum over finitely many points of functions going to zero almost surely and conclude that $\Delta^{(n)}\overset{\text{a.s.}}{\to}0$ and in turn that 
    \begin{align*}
        \lim_{n\to \infty}\sup_{x\in \amsmathbb{R}}|T^{(n)}(x)-T(x)|\leq \varepsilon \quad \text{a.s.}
    \end{align*}
    As $\varepsilon>0$ was arbitrary it holds that 
    \begin{align*}
       \lim_{n\to \infty}\sup_{x\in \amsmathbb{R}}|T^{(n)}(x)-T(x)|=0 \quad \text{a.s.}
    \end{align*} 
\end{proof}
\begin{lemma}[Hoeffding's inequality]\label{Hoeffding} Let $(X_{i})$ be a sequence of independent real-valued random variables such that $a_{i}\leq X_{i}\leq b_{i}$ almost surely. Let $S_{n}=X_{1}+\ldots+X_{n}$. Then 
\begin{align*}
    P\left(|S_{n}-\E\left[S_{n}\right]|\geq t \right)\leq 2 \exp\left(-\frac{2t^{2}}{\sum_{i=1}^{n}(b_{i}-a_{i})^{2}}\right).
\end{align*}
\end{lemma}
% \begin{align*}
%     \z\mapsto \E\left[1_{(Y_{1}\leq t)}|\boldsymbol{Z}_{1}=\z\right]=F(t|\z)
% \end{align*}
\begin{lemma}[Weak consistency of the Parzen--Rosenblatt and Priestley--Chao estimators]\label{weak_consistency_proof} Let $(Y_{i})$ be a sequence of iid real-valued random variables with distribution function $F$ and let $(\boldsymbol{Z}_{i})$ be a sequence of iid random variables in $\amsmathbb{R}^{k}$ with a density $g$ that is $C^{2}$. Assume that the function $ \z\mapsto \E\left[1_{(Y_{1}\leq t)}|\boldsymbol{Z}_{1}=\z\right]=F(t|\z)$ is $C^{2}$ for every $t$. Then it holds for the Parzen--Rosenblatt estimator $\g(\z)$ of $g(\z)$ and the Priestley--Chao estimator $\mathbb{m}^{(n)}(t;\z)$ of $m(t;\z)=F(t|\z)g(\z)$ that
   \begin{align*}
    \E\left[\g(\z)\right]&= g(\z)+O\left(\left(\max_{i,j}\B_{n}(i,j)\right)^{2}\right),\\
\textnormal{Var}\left[\g(\z)\right]&=O\left(\frac{1}{n|\B_{n}|}\right),
\\
\E\left[\mm(t;\z)\right]&=F(t|\z)g(\z)+O\left(\left(\max_{i,j}\B_{n}(i,j)\right)^{2}\right),\\
\textnormal{Var}\left[\mm(t;\z)\right]&=O\left(\frac{1}{n|\B_{n}|}\right).
\end{align*}
% In particular it holds that
% \begin{align*}
%     \g(\z) \overset{P}{\to} g(\z), \quad \mm(\z)\overset{P}{\to} F(t|\z)g(\z).
% \end{align*}
\end{lemma}
\begin{proof}
    For this proof we follow calculations carried out in Chapter 1.4 in~\cite{KS}. Note that 
    \begin{align*}
        \E\left[\mm(t;\z)\right]&=\frac{1}{|\B_{n}|}\E\left[\E\left[1_{(Y_{1}\leq t)}K\left(\B_{n}^{-1}\left(\boldsymbol{Z}_{1}-\z\right)\right)\right]|\boldsymbol{Z}_{1}\right]
        \\& 
        =\frac{1}{|\B_{n}|}\int F(t|\boldsymbol{v})K\left(\B_{n}^{-1}\left(\boldsymbol{v}-\z\right)\right)g(\boldsymbol{v})\text{d}\boldsymbol{v}
        \\&
        =\int F\left(t|\z+\B_{n}\boldsymbol{u}\right)K(\boldsymbol{u})g\left(\z+\B_{n}\boldsymbol{u}\right)\mathrm{d}\uu,
    \end{align*}
    using the substitution $\boldsymbol{v}=\z+\boldsymbol{B}_{n}\boldsymbol{u}$ for the last equality. Taylor expanding around $\z$ yields
    \begin{align*}
        &\int \left(F(t|\z)+\left(\B_{n}\boldsymbol{u}\right)^{T}F'(t|\z)+\frac{1}{2}\left(\B_{n}\boldsymbol{u}\right)^{T}F''(t|\z)\left(\B_{n}\boldsymbol{u}\right)+o\left(||\left(\B_{n}\boldsymbol{u}\right)^{T}||^{2}\right)\right)
        \\&
        \cdot K(\boldsymbol{u})\left(g(\z)+\left(\B_{n}\boldsymbol{u}\right)^{T}g'(\z)+\frac{1}{2}\left(\B_{n}\boldsymbol{u}\right)^{T}g''(\z)\left(\B_{n}\boldsymbol{u}\right)+o\left(||\left(\B_{n}\boldsymbol{u}\right)^{T}||^{2}\right)\right)\text{d}\boldsymbol{u}
        \\&
        \quad = F(t|\z)g(\z)+O\left(\left(\max_{i,j}\boldsymbol{B}_{n}(i,j)\right)^{2}\right),
    \end{align*}
    showing that the Priestley--Chao estimator is asymptotically unbiased. For the variance we have 
    \begin{align*}
        &\text{Var}\left[\mm(t;\z)\right]=\frac{1}{n|\B_{n}|^{2}}\text{Var}\left[1_{(Y_{1}\leq t)}K\left(\B_{n}^{-1}\left(\boldsymbol{Z}_{1}-\z\right)\right)\right]
        \\
        &\quad =\frac{1}{n|\B_{n}|^{2}}\E\left[1_{(Y_{1}\leq t)}K^{2}\left(\B_{n}^{-1}\left(\boldsymbol{Z}_{1}-\z\right)\right)\right]\\
        &\quad\quad-\frac{1}{n|\B_{n}|^{2}}\E\left[1_{(Y_{1}\leq t)}K\left(\boldsymbol{B}_{n}^{-1}\left(\boldsymbol{Z}_{1}-\z\right)\right)\right]^{2}.
    \end{align*}
    Both terms are seen to be $O(1/(n|\B_{n}|))$ by the same procedure as above. The proof for the weak consistency of the Parzen--Rosenblatt estimator is virtually identical. 
\end{proof}
\begin{lemma}[Consistency of $\amsmathbb{H}^{(n)}_{\dagger,1}(\cdot|\z)$] \label{consistency_lemma}
%[Strong uniform consistency of the conditional crude expert Kaplan-Meier estimator]
Consider the setup of the conditional expert Kaplan--Meier estimator in Definition~\ref{def_expert_conditional_km}. Let Assumptions~\ref{Strong_uniform_assumptions_2} and the error term condition in~\eqref{consistency_assumption} hold. Then it holds that
% , let $g$ be $C^{2}$ and assume that the function $z\mapsto H_{1}(t|\z)$
% is $C^{2}$ for every $t\geq 0$. Assume that conditional right-censoring is entirely random and that
    % \begin{align}
    %     \bigg\vert\E\left[\md(t;\boldsymbol{z})\right]-H_{1}(t|\z)g(\z)\bigg\vert \to 0 \label{consistency1_assumption}
    % \end{align}
    % for every $t\geq 0$. Then it holds that 
    \begin{align*}
        \sup_{0\leq t <\infty}|\HHd(t|\boldsymbol{z})-H_{1}(t|\boldsymbol{z})|\overset{\textnormal{a.s.}}{\to}0.
    \end{align*}
\end{lemma}
\begin{proof}
%[\textnormal{\textbf{Proof of Lemma \ref{consistency_lemma}} (Consistency)}] 
    Note that 
    \begin{align*}
        &|\HHd(t|\boldsymbol{z})-H_{1}(t|\boldsymbol{z})|
        \\&\quad 
        \leq \Bigg \vert\frac{\md(t;\boldsymbol{z})}{\g(\boldsymbol{z})}-\frac{\E\left[\md(t;\boldsymbol{z})\right]}{\E\left[\g(\boldsymbol{z})\right]} \Bigg \vert +\Bigg \vert \frac{\E\left[\md(t;\boldsymbol{z})\right]}{\E\left[\g(\boldsymbol{z})\right]}-\frac{H_{1}(t|\boldsymbol{z})g(\boldsymbol{z})}{g(\boldsymbol{z})}\Bigg \vert.
    \end{align*}
    The second term goes to 0 pointwise in $t$ by the assumption in~\eqref{consistency_assumption} and by the weak consistency of $\g(\z)$, confer with Lemma~\ref{weak_consistency_proof} above. For the first term note that
    \begin{align*}
        &\Bigg \vert\frac{\md(t;\boldsymbol{z})}{\g(\boldsymbol{z})}-\frac{\E\left[\md(t;\boldsymbol{z})\right]}{\E\left[\g(\boldsymbol{z})\right]} \Bigg \vert
        \\&
        \leq \frac{\md(t;\boldsymbol{z})}{\E\left[\g(\boldsymbol{z})\right]\g(\boldsymbol{z})}\bigg \vert\E\left[\g(\boldsymbol{z})\right]-\g(\boldsymbol{z})\bigg\vert\\
        &\quad+\frac{1}{\E\left[\g(\boldsymbol{z})\right]}\bigg \vert \md(t;\boldsymbol{z})-\E\left[\md(t;\boldsymbol{z})\right]\bigg \vert.
    \end{align*}
    Using Hoeffding's inequality Lemma~\ref{Hoeffding} and the assumption in~\eqref{consistency2_assumptions} we obtain for the first term
    \begin{align*}
        \sum_{n=1}^{\infty}P\left(\bigg \vert \E\left[\g(\boldsymbol{z})\right]-\g(\boldsymbol{z})\bigg \vert >\varepsilon\right)\leq \sum_{n=1}^{\infty}2\exp\left(-\frac{2\varepsilon^{2}n|\B_{n}|^{2}}{\sup_{\uu}K^{2}(\uu)}\right)< \infty.
    \end{align*}
    By the Borel--Cantelli lemma it then holds that $|\E[\g(\boldsymbol{z})]-\g(\boldsymbol{z})|\overset{\text{a.s.}}{\to}0$. From this it follows that 
    \begin{align*}
        \bigg \vert \g(\z)-g(\z)\bigg \vert \leq \bigg \vert \g(\z)-\E\left[\g(\z)\right]\bigg \vert +\bigg \vert\E\left[\g(\z)\right]-g(\z)\bigg \vert\overset{\text{a.s.}}{\to}0
    \end{align*}
    and as $\md(\cdot;\z)$ is bounded we have by combination of limits that 
    \begin{align*}
        \frac{\md(t;\boldsymbol{z})}{\E\left[\g(\boldsymbol{z})\right]\g(\boldsymbol{z})}\bigg \vert\E\left[\g(\boldsymbol{z})\right]-\g(\boldsymbol{z})\bigg\vert\overset{\text{a.s.}}{\to}0.
    \end{align*}
    Repeating virtually the same arguments we conclude also that 
    \begin{align*}
        \frac{1}{\E\left[\g(\z)\right]}\bigg \vert \md(t;\boldsymbol{z})-\E\left[\md(t;\boldsymbol{z})\right]\bigg \vert\overset{\text{a.s.}}{\to}0.
    \end{align*}
    % 0 almost surely (use berstein). For the second term, let $\varepsilon>0$ be given and use Bernstein's inequality Lemma \ref{Berstein} to obtain
    % \begin{align*}
    %     &P\left(\bigg \vert \md(t;\boldsymbol{z})-\E\left[\md(t;\boldsymbol{z})\right]\bigg \vert >\varepsilon\right)
    %     \\&
    %     =P\left(n|\boldsymbol{B}_{n}|\bigg \vert \md(t;\boldsymbol{z})-\E\left[\md(t;\boldsymbol{z})\right] \bigg\vert >n|\boldsymbol{B}_{n}|\varepsilon\right)
    %     \\&
    %     \leq 2\exp\left(-\frac{1}{2}\frac{\left(n|\boldsymbol{B}_{n}|\varepsilon\right)^{2}}{nv+M\left(n|\boldsymbol{B}_{n}|\varepsilon\right)}\right)
    %     \\&
    %     =2\exp\left(-\frac{1}{2}\frac{n|\boldsymbol{B}_{n}|^{2}\varepsilon^{2}}{v+M|\boldsymbol{B}_{n}|\varepsilon}\right)
    %     \\&
    %     \leq 2\exp\left(-\frac{1}{2}\frac{n|\boldsymbol{B}_{n}|^{2}\varepsilon^{2}}{\Tilde{M}}\right)\to 0,
    % \end{align*}
    % by assumption. (the series of the above exponential should be assumed to be finite). We conclude that the following series is convergent 
    % \begin{align*}
    %     \sum_{n=1}^{\infty}P\left(\bigg \vert \md(t;\boldsymbol{z})-\E\left[\md(t;\boldsymbol{z})\right]\bigg \vert >\varepsilon\right)<\infty.
    % \end{align*}
    % Using the Borel-Cantelli lemma and that $\varepsilon>0$ was arbitrary we conclude in turn that 
    % \begin{align*}
    %     P\left(\bigg \vert \md(t;\boldsymbol{z})-\E\left[\md(t;\boldsymbol{z})\right]\bigg \vert \to 0\right)=1.
    % \end{align*}
    Collecting results gives 
    \begin{align*}
        |\HHd(t|\boldsymbol{z})-H_{1}(t|\boldsymbol{z})|\overset{\text{a.s.}}{\to}0.
    \end{align*}
    An application of Lemma~\ref{cadlag_bounded_as} finishes the proof. 
    % and as $\amsmathbb{H}_{\dagger,1}^{(n)}(\cdot|\z)$ is càdlàg and non-decreasing it follows by an application of Lemma \ref{cadlag_bounded_as} that
    % \begin{align*}
    %     \sup_{0\leq t <\infty}|\HHd(t|\boldsymbol{z})-H_{1}(t|\boldsymbol{z})| \overset{\text{a.s.}}{\to}0.
    % \end{align*}
\end{proof}
Next we give the proof of consistency of the conditional expert Kaplan--Meier estimator:
\begin{proof}[\textnormal{\textbf{Proof of Theorem \ref{consistency_theorem}} (Consistency of the expert estimator).}] Under Assumptions~\ref{Strong_uniform_assumptions_2} it holds for every $t\geq 0$ that $|\amsmathbb{H}^{(n)}(t|\z)-H(t|\z)|\overset{\text{a.s.}}{\to} 0$, confer with~\cite{Stute}. By Lemma~\ref{cadlag_bounded_as} this is readily expanded into functional convergence. Under conditional entirely random contaminated right-censoring, invoke the representation of $\Lambda(\cdot|\z)$ as in~\eqref{cumulative} and consider
    \begin{align*}
        |\mathbb{\Lambda}_{\dagger}^{(n)}(t|\z)-\Lambda(t|\z)|
        &
        \leq \bigg \vert \frac{1}{1-\HH(s-|\z)}-\frac{1}{1-H(s-|\z)}\text{d}\HHd(s|\z)\bigg \vert 
        \\&\quad +\bigg \vert \int_{[0,t]}\frac{1}{1-H(s-|\z)}\text{d}\left(\HHd(s|\z)-H_{1}(s|\z)\right)\bigg \vert.
    \end{align*}
The first term goes to zero almost surely on the interval $0\leq t\leq \theta <H^{-1}(1|\z)$ by
    \begin{align*}
        &\bigg \vert \int_{[0,t]}\frac{1}{1-\HH(s-|\z)}-\frac{1}{1-H(s-|\z)}\text{d}\HHd(s|\z)\bigg \vert 
        \\&\quad \leq \sup_{0\leq t \leq \theta}\bigg \vert \frac{1}{1-\HH(t-|\z)}-\frac{1}{1-H(t-|\z)}\bigg \vert\overset{\text{a.s.}}{\to} 0,
    \end{align*}
    which follows by the functional consistency of $\amsmathbb{H}^{(n)}(\cdot|\z)$. 
    The second term goes to zero almost surely by the following inequality and an application of Lemma~\ref{consistency_lemma} 
    \begin{align*}
        &\bigg \vert \int_{[0,t]}\frac{1}{1-H(s-|\z)}\text{d}\left(\HHd(s|\z)-H_{1}(s|\z)\right)\bigg \vert 
        \\&
        \quad \leq \frac{1}{1-H(\theta|\z)}\sup_{0\leq t \leq \theta}|\HHd(t|\z)-H_{1}(t|\z)|\overset{\text{a.s.}}{\to}0.
    \end{align*}
    Collecting results gives $\sup_{0\leq t \leq \theta}|\mathbb{\Lambda}_{\dagger}^{(n)}(t|\z)-\Lambda(t|\z)|\overset{\text{a.s.}}{\to}0$. From Proposition 7.2.1 in~\cite{EMP} it holds pathwise that 
    \begin{align*}
        &|\Fd(t|\z)-F(t|\z)|\\
        &=\bigg \vert (1-F(t|\z))\int_{[0,t]}\frac{1-\Fd(s-|\z)}{1-F(s|\z)}\text{d}\left(\mathbb{\Lambda}^{(n)}_{\dagger}(s|\z)-\Lambda(s|\z)\right)\bigg \vert,
    \end{align*}
    for $0\leq t < F^{-1}(1|\z)$. Since $\theta< F^{-1}(1|\z)$ we have that 
    \begin{align*}
        &|\Fd(t|\z)-F(t|\z)|\\
        &=\bigg \vert (1-F(t|\z))\int_{[0,t]}\frac{1-\Fd(s-|\z)}{1-F(s|\z)}\text{d}\left(\mathbb{\Lambda}^{(n)}_{\dagger}(s|\z)-\Lambda(s|\z)\right)\bigg \vert
        \\&\quad \leq \frac{1}{1-F(\theta|\z)}\bigg \vert \int_{[0,t]}\text{d}\left(\mathbb{\Lambda}_{\dagger}^{(n)}(s|\z)-\Lambda(s|\z)\right)\bigg \vert 
        \\&\quad \leq \frac{1}{1-F(\theta|\z)}\sup_{0\leq t \leq \theta}|\mathbb{\Lambda}_{\dagger}^{(n)}(t|\z)-\Lambda(t|\z)|\overset{\text{a.s.}}{\to}0.
    \end{align*}
\end{proof}
\textnormal{\textbf{Justification of Remark~\ref{bandwidth_remark_consistency}}.}
\label{bandwith_fulfills}
    The bandwidth matrix readily satisfies conditions~\eqref{strong_uniform_consistency_2_i} and~\eqref{strong_uniform_consistency_2_ii} in Assumptions~\ref{Strong_uniform_assumptions_2} and the condition~\eqref{strong_uniform_consistency_2_iii} is for $r=2$ fulfilled by  
    \begin{align*}
    \sum_{n=1}^{\infty}c_{n}^{r}=\sum_{n=1}^{\infty}\left(\frac{1}{n}\right)^{(1-\rho)r}<\infty.
    \end{align*}
    For the condition in~\eqref{strong_uniform_consistency_2_iv}, let $M(\varepsilon)=2\varepsilon^{2} / \sup_{\boldsymbol{u}}K^{2}(\boldsymbol{u})$ to obtain
    \begin{align*}
        &\sum_{n=1}^{\infty}\exp\left(-\frac{2\varepsilon^{2} n|\B_{n}|^{2}}{\sup_{\uu}K^{2}(\uu)}\right)=\sum_{n=1}^{\infty}\exp\left(-M(\varepsilon)\frac{\log^{2}(n)}{n^{2\rho-1}}\right)
        \\&\quad \leq \sum_{n=1}^{\infty}\exp\left(-M(\varepsilon)\frac{\log(n)}{n^{2\rho-1}}\right)=\sum_{n=1}^{\infty}\left(\frac{1}{n}\right)^{M(\varepsilon)n^{1-2\rho}}<\infty.
    \end{align*}  
    The expert meets the quality condition in \eqref{consistency_assumption} by
    \begin{align*}
    &\bigg\vert\E\left[\md(t;\boldsymbol{z})\right]-H_{1}(t|\z)g(\z)\bigg\vert
    \\&\leq \bigg\vert\E\left[\md(t;\boldsymbol{z})\right]-\E\left[\sum_{i=1}^{n}1_{(W_{i}\leq t,W_{i}=X_{i})}K_{\B_{n}}\left(\Z_{i}-\z\right)\right]\bigg \vert 
    \\&\quad +\bigg\vert\E\left[\sum_{i=1}^{n}1_{(W_{i}\leq t,W_{i}=X_{i})}K_{\B_{n}}\left(\Z_{i}-\z\right)\right]-H_{1}(t|\z)g(\z)\bigg \vert.
    \end{align*}
    The second term converges to 0 by Lemma~\ref{weak_consistency_proof}. To deal with the first term, note that
    \begin{align*}
        &P(W\leq t, \eta_{1}=1|\boldsymbol{Z}=\z)-H_{1}(t|\z)
        \\&\quad =\E\left[1_{(W\leq t)}(\eta_{1}-1_{(W=X)})|\boldsymbol{Z}=\z\right]
        \\&\quad =P(\delta=1|\boldsymbol{Z}=\z)\E\left[1_{(W\leq t)}\left(\eta_{1}-1_{(W=X)}\right)|\delta=1,\boldsymbol{Z}=\z\right]
        \\&
        \quad=P(\delta=1|\boldsymbol{Z}=\z)
        \\&\quad \quad \cdot\int_{[0,t]}\E\left[\eta_{1}-1_{(W=X)}|\delta=1,W=w,\boldsymbol{Z}=\z\right]\text{d}P(W\leq w|\delta=1,\boldsymbol{Z}=\z)
        \\&\quad =\int_{[0,t]}\E\left[\eta_{1}-1_{(W=X)}|\delta=1,W=w,\boldsymbol{Z}=\z\right]\text{d}H_{1}^{\times}(w|\z)
        \\&
        \quad = 0,
    \end{align*}
where $H_{1}^{\times}(t|\z)=P(W\leq t,\delta=1|\boldsymbol{Z}=\z)$. The expert event distribution then fulfills the $C^{2}$-property and applications of Lemma~\ref{weak_consistency_proof} gives
    \begin{align*}
       &\bigg\vert\E\left[\md(t;\boldsymbol{z})\right]-\E\left[\sum_{i=1}^{n}1_{(W_{i}\leq t,W_{i}=X_{i})}K_{\B_{n}}\left(\Z_{i}-\z\right)\right]\bigg \vert = O(b_{n}^{2})\to 0.
    %    \\& \quad = \Big \vert g(\z)\E\left[1_{(W\leq t)}(\eta_{1}-1_{(W=X)})|\boldsymbol{Z}=\z\right]\Big \vert +O(b_{n}^{2})
    %     \\&
    %    \quad = \Big \vert g(\z)P(\delta=1|\boldsymbol{Z}=\z)\E\left[1_{(W\leq t)}\left(\eta_{1}-1_{(W=X)}\right)|\delta=1,\boldsymbol{Z}=\z\right]\Big \vert +O(b_{n}^{2})
    %    \\&
    %    \quad = \Big \vert g(\z)P(\delta=1|\boldsymbol{Z}=\z)
    %    \\&\quad \quad \cdot \int_{[0,t]}\E\left[\eta_{1}-1_{(W=X)}|\delta=1,W=w,\boldsymbol{Z}=\z\right]\text{d}P(W\leq w|\delta=1,\boldsymbol{Z}=\z)\Big \vert +O(b_{n}^{2})
    %    \\&
    %    \quad = \Big \vert g(\z)\int_{[0,t]}\E\left[\eta_{1}-1_{(W=X)}|\delta=1,W=w,\boldsymbol{Z}=\z\right]\text{d}H_{1}^{\times}(w|\z)\Big \vert +O(b_{n}^{2})
    % \\&\quad =0+O(b_{n}^{2})\to 0,
    \end{align*} 
    This finishes the justification.
\begin{proof}[\textnormal{\textbf{Proof of Corollary \ref{bias}} (Consistency bias of the expert estimator).}] Follow the proof of Lemma~\ref{consistency_lemma} to obtain that
    \begin{align*}
        \bigg \vert \HHd(t|\z)-\left(H_{1}(t|\z)+\frac{\gamma(t;\z)}{g(\z)}\right)\bigg \vert 
        &\leq \Bigg \vert\frac{\md(t;\boldsymbol{z})}{\g(\boldsymbol{z})}-\frac{\E\left[\md(t;\boldsymbol{z})\right]}{\E\left[\g(\boldsymbol{z})\right]} \Bigg \vert 
        \\&\quad +\Bigg \vert \frac{\E\left[\md(t;\boldsymbol{z})\right]}{\E\left[\g(\boldsymbol{z})\right]}-\frac{H_{1}(t|\boldsymbol{z})g(\boldsymbol{z})+\gamma(t;\z)}{g(\boldsymbol{z})}\Bigg \vert.
    \end{align*}
    The first term is known to converge to 0 almost surely by the proof of Lemma~\ref{consistency_lemma} and the second term also converges to 0 due to the bias assumption. It then holds that $\sup_{0\leq t < \infty} | \HHd(t|\z)-\left(H_{1}(t|\z)+\gamma(t;\z)/g(\z)\right)|\overset{\text{a.s.}}{\to}0$ and in turn the bias for the conditional expert Nelson--Aalen estimator is for $0\leq t \leq \theta<H^{-1}(1|\z)$ given by $\Gamma(\cdot;\z)$. This is seen by
  %   \begin{align*}
  %       \Gamma(t;\z)=\int_{[0,t]}\frac{1}{g(\z)(1-H(s-|\z))}\text{d}\gamma(s;\z)
  %   \end{align*}
  % which is seen by
    \begin{align*}
        &\sup_{0\leq t \leq \theta}\bigg \vert \mathbb{\Lambda}_{\dagger}^{(n)}(t|\z)-\left(\Lambda(t|\z)+\Gamma(t;\z)\right)\bigg \vert 
        \\&\leq \sup_{0\leq t \leq \theta}\bigg \vert \frac{1}{1-\HH(s-|\z)}-\frac{1}{1-H(s-|\z)}\text{d}\HHd(s|\z)\bigg \vert 
        \\&\quad +\sup_{0\leq t \leq \theta}\bigg \vert \int_{[0,t]}\frac{1}{1-H(s-|\z)}\text{d}\left(\HHd(s|\z)-(H_{1}(s|\z)+\gamma(s;\z)/g(\z))\right)\bigg \vert,
    \end{align*}
    which is almost surely convergent to 0 by virtually the same arguments as in the proof of Theorem~\ref{consistency_theorem} above. Using this and utilizing the Duhamel equation, confer with Chapter 3.10.5.5 in~\cite{WCEP}, yields for the biased conditional expert Kaplan--Meier estimator 
    \begin{align*}
        &\sup_{0\leq t \leq \theta}|\Fd(t|\z)-\phi\left(-(\Lambda(t|\z)+\Gamma(t;\z))\right)|
        % \\& =\sup_{0\leq t \leq \theta}|\phi(-\mathbb{\Lambda}_{\dagger}\left(t|\z)\right)-\phi\left(-(\Lambda(t|\z)+\Gamma(t;\z))\right)|
        \\&
        = \sup_{0\leq t \leq \theta} \biggl\{ \phi\left(-\mathbb{\Lambda}^{(n)}_{\dagger}(t|\z)\right)
        \\&
        \quad \cdot \int_{[0,t]}\frac{\phi\left(-\left(\Lambda(s-|\z)+\Gamma(s-;\z)\right)\right)}{\phi\left(-\mathbb{\Lambda}^{(n)}_{\dagger}(s|\z)\right)}\text{d}\left(-\mathbb{\Lambda}^{(n)}_{\dagger}(s|\z)+\left(\Lambda(s|\z)+\Gamma(s;\z)\right)\right)\biggr\}
        \\& \leq \frac{\sup_{0\leq t \leq \theta}\phi\left(-\left(\Lambda(t-|\z)+\Gamma(t-;\z)\right)\right)}{\phi\left(-\mathbb{\Lambda}^{(n)}_{\dagger}\left(\theta|\z\right)\right)}\\
        &\quad\times \sup_{0\leq t \leq \theta}|-\mathbb{\Lambda}^{(n)}_{\dagger}(t|\z)+\left(\Lambda(t|\z)+\Gamma(t;\z)\right)|\overset{\text{a.s.}}{\to}0.
    \end{align*}
    % In the case of continuity of the bias and the true distributions $H(\cdot|\z)$ and $H_{1}(\cdot|\z)$ we have by multiplicity of the product integral that 
    % \begin{align*}
    %     |\Fd(t|\z)-\phi\left(\Lambda(t|\z)-\Gamma(t;\z)\right)|=|\Fd(t|\z)-F(t|\z)\text{e}^{-\Gamma(t;\z)}|.
    % \end{align*}
    % Assuming also that $t\mapsto \gamma(t;\z)=\gamma(\z)$ is a constant we have 
    % \begin{align*}
    % \Gamma(t;\z)=\gamma(\z)/f(\z)\Delta(1-H(t|\z))^{-1}.
    % \end{align*}
\end{proof}

\section{Weak convergence}\label{appB}
Here, we ultimately give proofs of weak convergence Theorem~\ref{weakconvergence} and Corollary~\ref{bias_weakconvergence} concerning biased experts. Lemma~\ref{weakconvergence_lemma} gives convergence in fidis for the Nadaraya--Watson estimator and Lemma~\ref{weak_convergence_lemma} gives functional weak convergence for the Nadaraya--Watson estimator incorporating expert judgments.
% \begin{align*}
%     \mm(\z)=\frac{1}{n|\B_{n}|}\sum_{m=1}^{n}
% \end{align*}
% \begin{align}\label{lemma_weakconverence_assumption}
%     \sqrt{n|\B_{n}|}\left(\max_{i,j}\B_{n}(i,j)\right)^{2}\to 0.
% \end{align}
\begin{lemma}[Weak convergence] \label{weakconvergence_lemma} Let $(Y_{i})$ be a sequence of iid real-valued random variables with distribution function $F$ and let $(\boldsymbol{Z}_{i})$ be a sequence of iid random variables in $\amsmathbb{R}^{k}$ with a density $g$ that is $C^{2}$. Let $\boldsymbol{t}_{l}\in \amsmathbb{R}^{l}$ and assume that the vector function $\boldsymbol{z}\mapsto \amsmathbb{E}\left[(1_{(Y_{1}\leq t_{1})},\dots,1_{(Y_{l}\leq t)})|\boldsymbol{Z}_{1}=\z\right]=F(\boldsymbol{t}_{l}|\z)$ is entry-wise $C^{2}$ and assume for the bandwidth that $\sqrt{n|\B_{n}|}(\max_{i,j}\B_{n}(i,j))^{2}\to 0$.
Then it holds that the Nadaraya--Watson estimator $\amsmathbb{F}^{(n)}(\boldsymbol{t}_{l}|\z)=\mathbb{m}^{(n)}(\boldsymbol{t}_{l};\z)/\mathbb{g}(\z)$ of $F(\boldsymbol{t}_{l}|\z)$ is weakly convergent 
\begin{align*}
    \sqrt{n|\B_{n}|}\left(\amsmathbb{F}^{(n)}(\boldsymbol{t}_{l}|\z)-
    F(\boldsymbol{t}_{l}|\z)\right)\overset{\textnormal{D}}{\to}\mathcal{N}\left(\boldsymbol{0},\Sigma\right),
\end{align*}
and for $1\leq i,j\leq l$ the covariance is given by
\begin{align*}
    \Sigma(i,j)=\frac{1}{g(\z)}\int K^{2}(\uu)\textnormal{d}\uu\left(F(t_{i}\wedge t_{j}|\z)-F(t_{i}|\z)F(t_{j}|\z)\right).
\end{align*}
\end{lemma}
\begin{proof}
    First note that 
    \begin{align*} &\sqrt{n|\B_{n}|}\left(\amsmathbb{F}^{(n)}(\boldsymbol{t}_{l}|\z)-F(\boldsymbol{t}_{l}|\z)\right)
    \\&
    \quad =\sqrt{n|\B_{n}|}\left(\frac{\mm(\boldsymbol{t}_{l};\boldsymbol{z})}{\g(\boldsymbol{z})}-\frac{\E\left[\mm(\boldsymbol{t}_{l};\boldsymbol{z})\right]}{\E\left[\g(\z)\right]}\right)\\
    &\quad\quad+\sqrt{n|\B_{n}|}\left(\frac{\E\left[\mm(\boldsymbol{t}_{l};\boldsymbol{z})\right]}{\E\left[\g(\z)\right]}-\frac{F(\boldsymbol{t}_{l}|\z)g(\z)}{g(\z)}\right).
    \end{align*}
    The second term converges to 0 by the weak consistency of the estimators and the bandwidth assumption. The first term can be expressed as
    \begin{align*}
        &\sqrt{n|\B_{n}|}\left(\frac{\mm(\boldsymbol{t}_{l};\boldsymbol{z})}{\g(\boldsymbol{z})}-\frac{\E\left[\mm(\boldsymbol{t}_{l};\boldsymbol{z})\right]}{\E\left[\g(\z)\right]}\right)
        \\& \quad = -\frac{\mm(\boldsymbol{t}_{l};\z)}{\E\left[\g(\z)\right]\g(\z)}\sqrt{n|\B_{n}|}\left(\g(\z)-\E\left[\g(\z)\right]\right) 
        \\&\quad \quad 
        + \frac{1}{\E\left[\g(\z)\right]}\sqrt{n|\B_{n}|}\left(\mm(\boldsymbol{t}_{l};\z)-\E\left[\mm(\boldsymbol{t}_{l};\z)\right]\right)
        \\&
        \quad = \boldsymbol{A}_{n}\sum_{m=1}^{n}\left(\boldsymbol{W}_{nm}-\E\left[\boldsymbol{W}_{nm}\right]\right),
    \end{align*}
    where the $\boldsymbol{A}_{n}$ is a matrix in $\amsmathbb{R}^{l\times (l+1)}$ and $\boldsymbol{W}_{nm}$ are vectors in $\amsmathbb{R}^{l+1}$ given by
    \begin{align*}
        \boldsymbol{A}_{n}(i,j)&= \begin{cases}
        \frac{1}{\E\left[\g(\z)\right]} & i=j,j\leq l 
        \\
        -\frac{\mathbb{m}^{(n)}(t_{i};\z)}{\E\left[\g(\z)\right]\g(\z)} & i\leq l, j=l+1
        \\
        0 & \text{otherwise},
        \end{cases}
        \\ 
        \boldsymbol{W}_{nm}(i)&=\begin{cases}
            \frac{1}{\sqrt{n|\B_{n}|}}1_{(Y_{nm}\leq t_{i})}K\left(\B_{n}^{-1}(\boldsymbol{Z}_{nm}-\z)\right) & i\leq l
            \\ 
            \frac{1}{\sqrt{n|\B_{n}|}}K\left(\B_{n}^{-1}(\boldsymbol{Z}_{nm}-\z)\right) & i = l+1.
        \end{cases}
    \end{align*}
    The vectors $(\boldsymbol{W}_{nm}-\E[\boldsymbol{W}_{nm}])$ for $n\in \amsmathbb{N}, m \leq n$ are a triangular array of centered random variables with finite second moment and it holds that 
    \begin{align*}
        \sum_{m=1}^{n}\text{Var}\left[\boldsymbol{W}_{nm}\right]\to \begin{cases}
             F(t_{i}\wedge t_{j}|\z)g(\z)\int K^{2}(\boldsymbol{u})\text{d}\boldsymbol{u} & i\leq l, j \leq l 
             \\
             F(t_{i}|\z)g(\z)\int K^{2}(\boldsymbol{u})\text{d}\boldsymbol{u} & i\leq l, j=l+1 %\text{ and } i=l+1, j\leq l
             \\
             g(\z)\int K^{2}(\boldsymbol{u})\text{d}\boldsymbol{u} & i=j=l+1
        \end{cases}
    \end{align*}
     as essentially calculated in Lemma~\ref{weak_consistency_proof}. Lyapunov's condition for Lindeberg's CLT is for $\alpha=3$ verified by 
     \begin{align*}
         &\sum_{m=1}^{n}\E\left[||\boldsymbol{W}_{nm}-\E\left[\boldsymbol{W}_{nm}\right]||^{3}\right]
         \\& 
         \quad \leq \sum_{m=1}^{n}\E\left[||\boldsymbol{W}_{nm}||^{3}\right]
         % \\&\quad \leq \sum_{m=1}^{n}\frac{1}{n^{3/2}|\B_{n}|^{3/2}}\E\left[||\left(
         %      C K\left(\B_{n}^{-1}(\boldsymbol{Z}_{nm}-\z)\right),..., CK\left(\B_{n}^{-1}(\boldsymbol{Z}_{nm}-\z)\right)
         % \right)^{T}||^{3}\right]
         \\&
         \quad \leq \sum_{m=1}^{n}\frac{1}{n^{3/2}|\B_{n}|^{3/2}}\E\left[k^{3}K^{3}\left(\B_{n}^{-1}(\boldsymbol{Z}_{nm}-\z)\right)\right]
         \\&
         \quad \leq \frac{k^{3}}{n^{1/2}|\B_{n}|^{1/2}}\left(g(\z)\int K^{3}(\uu)\mathrm{d}\uu+o(1)\right) \to 0.
     \end{align*}
     As $\boldsymbol{A}_{n}$ converges in probability the proof is finished.     
\end{proof}
\begin{lemma}[Weak convergence of $\amsmathbb{H}^{(n)}_{\dagger,1}(\cdot|\z)$] \label{weak_convergence_lemma}Let Assumptions~\ref{weak_convergence_assumptions_2} hold and assume that
\begin{align}
    \label{weak_convergence_lemma_condition}\sqrt{n|\boldsymbol{B}_{n}|}\bigg\vert\E\left[\md(t;\boldsymbol{z})\right]-H_{1}(t|\z)g(\z)\bigg\vert \to 0 
\end{align}
for every $t\geq 0$. Then it holds that
\begin{align*}
   \sqrt{n|\B_{n}|}\left(\HHd(\cdot|\z)-H_{1}(\cdot|\z)\right)\overset{\textnormal{D}}{\to}\G_{1}(\cdot|\boldsymbol{z})\quad \text{in }\ell^{\infty}\left([0,\infty)\right),
\end{align*}
where $\G_{1}(\cdot|\z)$ is a zero-mean Gaussian process with covariance function 
\begin{align*}
    \sigma^{2}_{\G_{1}(\cdot|\z)}(t,s)=\left(H_{1}(t\wedge s|\boldsymbol{z})-H_{1}(t|\boldsymbol{z})H_{1}( s|\boldsymbol{z})\right)\frac{1}{g(\z)}\int K^{2}(\boldsymbol{u})\textnormal{d}\boldsymbol{u}.
\end{align*}   
\end{lemma}
\begin{proof}
Consider first the process
\begin{align*}
 &\sqrt{n|\boldsymbol{B}_{n}|}\left(\md(\cdot ;\boldsymbol{z})-H_{1}(\cdot|\boldsymbol{z})g(\boldsymbol{z})\right)
 \\&
 \quad =\sqrt{n|\boldsymbol{B}_{n}|}\left(\md(\cdot ;\boldsymbol{z}) -\E\left[\md(\cdot;\boldsymbol{z})\right]\right)\\
 &\quad\quad-\sqrt{n|\boldsymbol{B}_{n}|}\left(H_{1}(\cdot |\boldsymbol{z})g(\boldsymbol{z})-\E\left[\md(\cdot;\boldsymbol{z})\right]\right)
\end{align*}
and note that the second term converges to 0 in $\ell^{\infty}([0,\infty))$ by the assumption in~\eqref{weak_convergence_lemma_condition} and Lemma~\ref{cadlag_bounded_as}, which applies also in absence of stochastic convergence.
The fidis of the first term exists by Lemma~\ref{weakconvergence_lemma} and its functional weak convergence is then found by an application of the bracketing CLT for triangular arrays, see Theorem 2.11.9 in~\cite{WCEP}: Let
    \begin{align*}
        \mathcal{F}_{1}=\left\{f:[0,\infty)^{2}\to [0,\infty) \text{ s.t. }f(x,y)=1_{[0,t]}(x)y \text{ for } t\in [0,\infty)\right\}
    \end{align*}
    and define $A_{nm}:\mathcal{F}_{1}\to \R$ as the mapping 
    \begin{align*}
        A_{nm}(f)&=f\left(W_{nm},\eta_{nm} K\left(\B_{n}^{-1}(\Z_{nm}-\z)\right)/\sqrt{n|\boldsymbol{B}_{n}|}\right)
        \\&
        =\frac{1}{\sqrt{n|\B_{n}|}}1_{(W_{nm}\leq t)}\eta_{nm}K\left(\B_{n}^{-1}(\Z_{nm}-\z)\right).
    \end{align*}
    For every $n$, $A_{n1},\ldots,A_{nn}$ are mutually independent stochastic processes indexed by $\mathcal{F}_{1}$. The 1st condition is fulfilled as it holds for every $c>0$ that
    \begin{align*}
        &\sum_{m=1}^{n}\E^{*}\left[||A_{nm}||_{\mathcal{F}_{1}}1_{\left(||A_{nm}||_{\mathcal{F}_{1}>c}\right)}\right]
        \\&\quad \leq \sum_{m=1}^{n}\E^{*}\left[||A_{nm}||_{\mathcal{F}_{1}}\left(\frac{||A_{nm}||_{\mathcal{F}_{1}}}{c}\right)^{2}\right]
        \\&\quad 
        \leq\sum_{m=1}^{n} \E\left[\frac{K^{3}\left(\B_{n}^{-1}(\Z_{nm}-\z)\right)}{n^{3/2}|\B_{n}|^{3/2}c^{2}}\right]
        \\&
        \quad =\frac{1}{ n^{1/2}|\B_{n}|^{1/2}c^{2}}\left(g(\z)\int K^{3}(\uu)\text{d}\uu+o(1)\right)\to 0.
        %\\&\leq\sup_{u}K(\boldsymbol{u})^{3}\sum_{m=1}^{n}%\frac{1}{n^{3/2}|\B_{n}|^{3/2}\eta^{3}}\to 0.
    \end{align*}
    For the 3rd condition we construct a partition of the index set $\mathcal{F}_{1}$ independent of the row number $n$ and bound the bracketing number $N_{[\:]}\left(\varepsilon,\mathcal{F}_{1},L_{2}^{n}\right)$. Consider a partition $0=t_{0}<t_{1}<\ldots<t_{k}=\infty$ and consider functions on the form $l_{j}(x,y)=1_{[0,t_{j-1}]}(x)y$ and $u_{j}(x,y)=1_{[0,t_{j})}(x)y$ for $j=1,\ldots,k$. Let 
    \begin{align*}
        \mathcal{F}_{\varepsilon,j}=\left\{f\in \mathcal{F}_{1}:l_{j}\leq f \leq u_{j}\right\}\quad \text{for }1\leq j \leq k.
    \end{align*}
    Then 
    \begin{align*}
        \mathcal{F}_{1}=\bigcup_{j=1}^{k}\mathcal{F}_{\varepsilon,j}.
    \end{align*}
    For every partitioning set $\mathcal{F}_{\varepsilon,j}$ it holds that 
    \begin{align*}
        &\sum_{m=1}^{n}\E^{*}\left[\sup_{f,g\in \mathcal{F}_{\varepsilon,j}}|A_{nm}(f)-A_{nm}(g)|^{2}\right]
        \\&
        =\frac{1}{n|\B_{n}|}\sum_{m=1}^{n}\E\left[\sup_{s,t\in [t_{j-1},t_{j})}|\left(1_{(W_{nm}\leq t)}-1_{(W_{nm}\leq s)}\right)\eta_{nm}K\left(\B_{n}^{-1}\left(\Z_{nm}-\z\right)\right)|^{2}\right]
        \\&
        \leq \frac{1}{n|\B_{n}|}\sum_{m=1}^{n}\E\left[\sup_{s,t\in [t_{j-1},t_{j})}|\left(1_{(W_{nm}\leq t)}-1_{(W_{nm}\leq s)}\right)K\left(\B_{n}^{-1}\left(\Z_{nm}-\z\right)\right)|^{2}\right]
        \\&
        =\left(H(t_{j}-|\z)-H(t_{j-1}|\z)\right)g(\z)\int K^{2}(\uu)\text{d}\uu+o(1).
    \end{align*}
    Note that the function $t\mapsto H(t|\z)$ is càdlàg, bounded and non-decreasing. Hence we may follow the arguments in Example 2.11.16 in~\cite{WCEP}: Given $\varepsilon>0$ choose $\Tilde{\varepsilon}$ such that $\Tilde{\varepsilon}\leq \varepsilon$ and choose a partition with less than $K/\Tilde{\varepsilon}^{2}$ terms, where $K>0$ is some constant, such that 
    \begin{align*}
        \left(H(t_{j}-|\z)-H(t_{j-1}|\z)\right)g(\z)\int K^{2}(\uu)\text{d}\uu\leq \Tilde{\varepsilon}^{2}
    \end{align*}
    for every $j$ in the partition. For large enough $n$ it then holds that
    \begin{align*}
        \E\left[\sup_{f,g\in \mathcal{F}_{\varepsilon,j}}|A_{nm}(f)-A_{nm}(g)|\right]\leq \Tilde{\varepsilon}^{2}\leq \varepsilon^{2}
    \end{align*}
    implying that the bracketing number is bounded by $K/\Tilde{\varepsilon}^{2}$ for $n$ large enough and we obtain 
    \begin{align*}
        \int_{(0,\delta_{n}]}\sqrt{\log N_{[\:]}\left(\varepsilon,\mathcal{F}_{1},L_{2}^{n}\right)}\text{d}\varepsilon\leq \int_{(0,\delta_{n}]}\sqrt{\log\left(\frac{K}{\Tilde{\varepsilon}^{2}}\right)}\text{d}\varepsilon\to 0,
    \end{align*}
    for every $\delta_{n}\to 0$. As the partitioning sets are chosen independently of $n$, the 2nd condition is unnecessary and we conclude that the process $\sum_{m=1}^{n}\left(A_{nm}-\E\left[A_{nm}\right]\right)$ converges to a zero-mean Gaussian process. Now consider the process
    \begin{align*}
        &\sqrt{n|\B_{n}|}\left(\HHd(\cdot|\z)-H_{1}(\cdot|\z)\right)
        \\&\quad 
        =\sqrt{n|\B_{n}|}\left(\HHd(\cdot|\z)-\frac{\E\left[\md(\cdot;\z)\right]}{\E\left[\g(\z)\right]}\right)\\
        &\quad\quad+ \sqrt{n|\B_{n}|}\left(\frac{\E\left[\md(\cdot;\z)\right]}{\E\left[\g(\z)\right]}-\frac{H_{1}(\cdot|\z)g(\z)}{g(\z)}\right)
    \end{align*}
    % \end{align*}
    and note that the second term converges to $0$ in $\ell^{\infty}([0,\infty))$. Then 
    \begin{align*}
         &\sqrt{n|\B_{n}|}\left(\HHd(\cdot|\z)-H_{1}(\cdot|\z)\right)
         \\
         &\quad = -\frac{\md(\cdot;\z)}{\E\left[\g(\z)\right]\g(\z)}\sqrt{n|\B_{n}|}\left(\g(\z)-\E\left[\g(\z)\right]\right)
        \\&\quad \quad +\frac{1}{\E\left[\g(\z)\right]}\sqrt{n|\B_{n}|}\left(\md(\cdot;\z)-\E\left[\md(\cdot;\z)\right]\right)+o(1).
    \end{align*}
    % The fidis of the sum of the two terms above exists by an application of Lemma \ref{weakconvergence_lemma} followed by an application of standard continuous mapping.
    % \textcolor{red}{Make more neat.} The bivariate process corresponding to these two terms are established as weakly convergent by first noting that $\md(\cdot;\z)$ and $\g(\z)$ are convergent in probability in $\ell^{\infty}([0,\infty))$ and that the joint fidis exists.  Then repeated applications of Exercise 1.5.3 in \cite{WCEP}
    Both terms consist of an empirical process that is convergent in $\ell^{\infty}([0,\infty))$ multiplied by a factor that is also convergent in probability in $\ell^{\infty}([0,\infty))$, hence the products are weakly convergent by the Continuous Mapping Theorem 1.3.6 in~\cite{WCEP}. As the joint fidis of the two terms exist and because the processes are asymptotically tight, it follows from Exercise 1.5.3, also in~\cite{WCEP}, and yet another application of functional continuous mapping, that the sum of the two terms is weakly convergent. We conclude that 
    \begin{align*}
        \sqrt{n|\B_{n}|}\left(\HHd(\cdot|\z)-H_{1}(\cdot|\z)\right) \overset{\text{D}}{\to}\amsmathbb{G}_{1}(\cdot|\z)\quad \text{in }\ell^{\infty}([0,\infty)).
    \end{align*}
\end{proof}
% and as in the proof Lemma \ref{weak_convergence_lemma} above it holds that 
% \begin{align*}
%     \sqrt{n|\boldsymbol{B}_{n}|}\left(\amsmathbb{H}^{(n)}(\cdot|\z)-H(\cdot|\z)\right)\overset{\text{D}}{\to}\amsmathbb{G}(\cdot|\z) \quad \text{in }\ell^{\infty}([0,\infty)),
% \end{align*}
% where $\G(\cdot|\z)$ is a zero-mean Gaussian Process with covariance function 
% \begin{align*}
%     \sigma^{2}_{\amsmathbb{G}(\cdot|\z)}(t,s)=\left(H(t\wedge s|\boldsymbol{z})-H(t|\boldsymbol{z})H( s|\boldsymbol{z})\right)\frac{1}{g(\z)}\int K^{2}(\boldsymbol{u})\textnormal{d}\boldsymbol{u}.
% \end{align*}   
% Since the limit processes $\amsmathbb{G}(\cdot|\z)$ and $\amsmathbb{G}_{1}(\cdot|\z)$ restricted to the interval $[0,\theta]$ are tight it follows that the corresponding empirical processes of $\amsmathbb{H}^{(n)}(\cdot|\z)$ and $\amsmathbb{H}_{\dagger,1}^{(n)}(\cdot|\z)$ are asymptotically tight. An application of Lemma \ref{weakconvergence_lemma} shows that the joint fidis of the empirical processes exist. This suffices for the bivariate process to be weakly convergent, i.e.
\begin{proof}[\textnormal{\textbf{Proof of Theorem~\ref{weakconvergence}} (Weak convergence of the expert estimator)}]
Following virtually the same argument as in the proof of Lemma~\ref{weak_convergence_lemma} above concerning applications of the bracketing CLT and bivariate weak convergence in $\ell^{\infty}([0,\theta])^{2}$, it holds that  
    \begin{align*}
&\left\{\sqrt{n|\B_{n}|}(\amsmathbb{H}^{(n)}(\cdot|\z)-H(\cdot|\z)),\sqrt{n|\B_{n}|}(\amsmathbb{H}_{\dagger,1}^{(n)}(\cdot|\z)-H_{1}(\cdot|\z))\right\}\\
&\quad\overset{\text{D}}{\to}\left\{\amsmathbb{G}(\cdot|\z), \amsmathbb{G}_{1}(\cdot|\z)\right\},
% \sqrt{n|\B_{n}|}\left(\mathbb{m}^{(n)}(\cdot|\z)-m(\cdot|\z)\right)\overset{\text{D}}{\to}, \quad \sqrt{n|\B_{n}|}\left(\mathbb{m}_{\dagger,1}^{(n)}(\cdot|\z)-m_{1}(\cdot|\z)\right)\overset{\text{D}}{\to}
\end{align*}
where $\amsmathbb{G}(\cdot|\z)$ is a zero-mean Gaussian process with variance 
\begin{align*}
    \sigma^{2}_{\amsmathbb{G}(\cdot|\z)}(t,s)=\left(H(t\wedge s)-H(t)H(s)\right)\frac{1}{g(\z)}\int K^{2}(\boldsymbol{u})\text{d}\boldsymbol{u}
\end{align*}
and the bivariate limiting process is a zero-mean Gaussian process with covariance
\begin{align*}
    \sigma^{2}_{\{\amsmathbb{G}(\cdot|\z),\amsmathbb{G}_{1}(\cdot|\z)\}}(t,s)=\left(H_{1}(t\wedge s)-H(t)H_{1}(s)\right)\frac{1}{g(\z)}\int K^{2}(\boldsymbol{u})\text{d}\boldsymbol{u}.
\end{align*}
    Following Example 3.10.20 and Example 3.10.33 in~\cite{WCEP}, showing weak convergence of the ordinary Nelson--Aalen and Kaplan--Meier estimator, we conclude by the functional delta-theorem that 
    \begin{align*}
        \sqrt{n|\B_{n}|}\left(\LLd(\cdot |\boldsymbol{z})-\Lambda(\cdot |\boldsymbol{z})\right)&\overset{\textnormal{D}}{\to}\amsmathbb{L}(\cdot|\boldsymbol{z}),\,\, 
        \sqrt{n|\B_{n}|}\left(\Fd(\cdot |\boldsymbol{z})-F(\cdot |\boldsymbol{z})\right)\overset{\textnormal{D}}{\to}\amsmathbb{Z}(\cdot |\boldsymbol{z})
    \end{align*}
    in $\ell^{\infty}([0,\theta])$. It remains to calculate the covariance functions: By Lemma 20.10 in~\cite{AS} we may evaluate $\amsmathbb{L}(\cdot|\z)$ as 
    \begin{align*}
        \amsmathbb{L}(t|\z)&=\frac{\amsmathbb{G}_{1}(t|\z)}{\Bar{H}(t|\z)}-\int_{[0,t]}\amsmathbb{G}_{1}(s-|\z)\text{d} \left(\frac{1}{\Bar{H}(s|\z)}\right)+\int_{[0,t]}\frac{\amsmathbb{G}(s-|\z)}{\Bar{H}^{2}(s|\z)}\text{d} H_{1}(s|\z).
    \end{align*}
    From this it follows that 
    \begin{align*}
        \text{Cov}\left[\amsmathbb{L}(t|\z),\amsmathbb{L}(s|\z)\right]=\frac{1}{g(\z)}\int K^{2}(\boldsymbol{u})\text{d}\boldsymbol{u}\text{Cov}\left[\Tilde{\amsmathbb{L}}(t|\z),\Tilde{\amsmathbb{L}}(s|\z)\right],
    \end{align*}
    where $\Tilde{\amsmathbb{L}}(\cdot|\z)$ is the same process as $\amsmathbb{L}(\cdot|\z)$ but with the difference that the covariance function for the underlying bivariate process $(\amsmathbb{G}(\cdot|\z), \amsmathbb{G}_{1}(\cdot|\z))$ is multiplied through by $g(\z)/\int K^{2}(\boldsymbol{u})\text{d}\boldsymbol{u}$  implying that $\Tilde{\amsmathbb{L}}(\cdot|\z)$ is known from the proof of the ordinary Nelson--Aalen estimator to be a Gaussian process with covariance function 
    \begin{align*}
        \text{Cov}\left[\Tilde{\amsmathbb{L}}(t|\z), \Tilde{\amsmathbb{L}}(s|\z)\right]=\int_{[0,s\wedge t]}\frac{1-\Delta \Lambda(u|\z)}{\Bar{H}(u|\z)}\text{d}\Lambda(u|\z).
    \end{align*}
    Hence we conclude that
    \begin{align*}
        \sigma^{2}_{\amsmathbb{L}(\cdot|\z)}=\frac{1}{g(\z)}\int K^{2}(\boldsymbol{u})\text{d}\boldsymbol{u}\int_{[0,s\wedge t]}\frac{1-\Delta \Lambda(u|\z)}{\Bar{H}(u|\z)}\text{d}\Lambda(u|\z).
    \end{align*}
    As it holds that 
    \begin{align*}
        &\amsmathbb{Z}(t|\z)=\Bar{F}(t|\z)\int_{[0,t]}\frac{1}{1-\Lambda(s|\z)}\text{d}\amsmathbb{L}(s|\z)\\
        &\quad\overset{\text{D}}{=}\frac{1}{g(\z)}\int K^{2}(\boldsymbol{u})\text{d}\boldsymbol{u}\Bar{F}(t|\z)\int_{[0,t]}\frac{1}{1-\Lambda(s|\z)}\text{d}\Tilde{\amsmathbb{L}}(s|\z)
    \end{align*}
    it follows that 
    \begin{align*}
         \sigma^{2}_{\amsmathbb{Z}(\cdot|\z)}(t,s)&=\frac{\Bar{F}(t|\boldsymbol{z})\Bar{F}(s|\boldsymbol{z})}{g(\boldsymbol{z})}\int K^{2}(\boldsymbol{u})\textnormal{d}\boldsymbol{u}\int_{[0,s\wedge t]}\frac{1}{\Bar{H}(u|\boldsymbol{z})\left(1-\Delta \Lambda(u|\boldsymbol{z})\right)}\textnormal{d}\Lambda(u|\boldsymbol{z}).
    \end{align*}
\end{proof}
\begin{proof}[\textnormal{\textbf{Proof of Corollary~\ref{bias_weakconvergence}} (Biased weak convergence of the expert estimator)}] The result follows by the proofs of Lemma~\ref{weak_convergence_lemma} and Theorem~\ref{weakconvergence} using the biased error term in~\eqref{weak_convergence_errorterm}.
\end{proof}

% Use the standard \LaTeX\ commands for headings in \verb|{appendix}|.
% Headings and other objects will be numbered automatically.
% \begin{equation}
% \mathcal{P}=(j_{k,1},j_{k,2},\dots,j_{k,m(k)}). \label{path}
% \end{equation}

% Sample of cross-reference to the formula (\ref{path}) in Appendix \ref{appB}.
\end{appendix}

%%%%%%%%%%%%%%%%%%%%%%%%%%%%%%%%%%%%%%%%%%%%%%
%% Acknowledgements                         %%
%% should be provided in the                %%
%% Acknowledgements section.                %%
%%%%%%%%%%%%%%%%%%%%%%%%%%%%%%%%%%%%%%%%%%%%%%
% \begin{acks}[Acknowledgments]
% The authors would like to thank the anonymous referees, an Associate
% Editor and the Editor for their constructive comments that improved the
% quality of this paper.
% \end{acks}

%%%%%%%%%%%%%%%%%%%%%%%%%%%%%%%%%%%%%%%%%%%%%%
%% Funding information, if any,             %%
%% should be provided in the                %%
%% funding section.                         %%
%%%%%%%%%%%%%%%%%%%%%%%%%%%%%%%%%%%%%%%%%%%%%%
\begin{funding}
The first author was supported by the Carlsberg Foundation, grant CF23-1096.
\end{funding}

\bibliographystyle{imsart-number} % Style BST file (imsart-number.bst or imsart-nameyear.bst)
\bibliography{main.bib}       % Bibliography file (usually '*.bib')

%% or include bibliography directly:
% \begin{thebibliography}{9}

% \bibitem{r1}
% \textsc{Billingsley, P.} (1999). \textit{Convergence of
% Probability Measures}, 2nd ed.
% Wiley, New York.
% \MR{1700749}

% \bibitem{r2}
% \textsc{Bourbaki, N.}  (1966). \textit{General Topology}  \textbf{1}.
% Addison--Wesley, Reading, MA.

% \bibitem{r3}
% \textsc{Ethier, S. N.} and \textsc{Kurtz, T. G.} (1985).
% \textit{Markov Processes: Characterization and Convergence}.
% Wiley, New York.
% \MR{838085}

% \bibitem{r4}
% \textsc{Prokhorov, Yu.} (1956).
% Convergence of random processes and limit theorems in probability
% theory. \textit{Theory  Probab.  Appl.}
% \textbf{1} 157--214.
% \MR{84896}
% \end{thebibliography}

\end{document}